\documentclass{article}

\usepackage{amssymb, amsmath, amsthm, mathtools, graphicx, enumerate, color}
\usepackage[vlined,linesnumberedhidden,noend]{algorithm2e}
\usepackage{cancel}
\usepackage{flexisym}
\usepackage{hyperref}
\usepackage{breqn}
\usepackage{url}

\newtheorem{theorem}{Theorem}
\newtheorem{remark}{Remark}
\newtheorem{lemma}{Lemma}
\newtheorem{proposition}{Proposition}
\newtheorem{corollary}{Corollary}

\newcommand{\size}[1]{\lvert#1\rvert}
\newcommand{\condi}[0]{\mid}

\newcommand{\st}[1]{\mathnormal{#1}}
\newcommand{\init}[0]{\st{o}}
\newcommand{\stateset}[0]{\st{Q}}

\newcommand{\prematchset}[0]{\st{F}}

\newcommand{\cState}[0]{q}

\newcommand{\cPos}[0]{p}
\DeclareMathOperator{\mnshft}{minshift}

\newcommand{\next}[0]{\boldsymbol\alpha}
\newcommand{\trans}[0]{\boldsymbol\delta}
\newcommand{\shift}[0]{\boldsymbol\gamma}

\newcommand{\alp}[0]{\mathcal{A}}

\newcommand{\piid}[0]{\pi}

\newcommand{\defi}[1]{\textit{#1}}

\newcommand{\matchine}[0]{\Gamma}

\newcommand{\prob}[0]{p}
\newcommand{\p}[0]{\mathrm{P}}
\newcommand{\Hstate}[1][H]{Q_{#1}}
\newcommand{\Hinit}[1][H]{\pi_{#1}}
\newcommand{\Htrans}[1][H]{\delta_{#1}}
\newcommand{\Hemis}[1][H]{\phi_{#1}}
\newcommand{\HemisBis}[1][H]{\psi_{#1}}

\newcommand{\model}[0]{\mathcal{M}}
\newcommand{\okw}[1][w]{$w$-consistent}
\newcommand{\expect}[0]{\mathbf{E}}

\newcommand{\tac}[1]{a_{#1}}
\newcommand{\as}[2]{\mathrm{A\!S}_{#1}(#2)}
\newcommand{\asC}[3]{\mathrm{A\!S}^{(#1)}_{#2}(#3)}
\newcommand{\alg}[0]{\mathbf{A}}
\newcommand{\proba}[1]{\mathrm{p}_{#1}}
\newcommand{\probaS}[1]{\mathrm{p}^{\star}_{#1}}
\newcommand{\Mstate}[1][M]{\mathcal{Q}_{#1}}
\newcommand{\Minit}[1][M]{\pi_{#1}}
\newcommand{\Mtrans}[1][M]{\delta_{#1}}

\newcommand{\ordmatchine}[1][\matchine]{O_{#1}}
\newcommand{\readset}[1]{R_{#1}}
\newcommand{\sink}{\odot}
\newcommand{\expan}[1]{\widehat{#1}}

\newcommand{\mem}[2][\matchine]{\mathbf{h}_{#1}(#2)}
\newcommand{\increas}{\phi}
\newcommand{\subinc}[2][\kappa]{S_{#1}(#2)}
\newcommand{\kshifted}[2]{\stackrel{\scriptscriptstyle#1}{\overleftarrow{#2}}}
\newcommand{\myarrow}[0]{{\scriptscriptstyle\triangleright}}
\newcommand{\redir}[3]{#1_{#3\myarrow#2}}

\newcommand{\remo}[2]{#1_{{\scriptsize\bcancel{#2}}}}
\newcommand{\freq}[0]{\alpha}

\newcommand{\curState}[3]{\mathbf{\cState}_{#1}^{#2}(#3)}
\newcommand{\curPos}[3]{\mathbf{\cPos}_{#1}^{#2}(#3)}
\newcommand{\curShift}[3]{\mathbf{s}_{#1}^{#2}(#3)}

\newcommand{\myshift}[0]{\phi}

\newcommand{\typa}[1]{\dot{#1}}
\newcommand{\typb}[1]{\ddot{#1}}
\newcommand{\sta}[1][q]{\typa{#1}}
\newcommand{\stb}[1][q]{\typb{#1}}

\newcommand{\fhsgn}{f}

\newcommand{\fhn}[4][n]{\fhsgn^{(#1)}_{#2}(#3,#4)}
\newcommand{\fh}[3]{\fhsgn_{#1}(#2,#3)}
\newcommand{\fon}[3][n]{g^{(#1)}_{#2}(#3)}
\newcommand{\fo}[2]{g_{#1}(#2)}

\newcommand{\fha}[3]{\typa{\fhsgn}_{#1}(#2,#3)}

\newcommand{\fhb}[3]{\typb{\fhsgn}_{#1}(#2,#3)}

\newcommand{\mm}[0]{M}
\newcommand{\modb}[0]{\typb{M}}
\newcommand{\moda}[0]{\typa{M}}

\newcommand{\set}[1]{\mathcal{#1}}
\newcommand{\subcase}[1]{\paragraph{\textit{#1}}}

\newcommand{\rvna}[1]{\typa{N}_{#1}}

\newcommand{\rvp}[1]{A_{#1}}
\newcommand{\rvpa}[1]{\typa{A}_{#1}}

\newcommand{\rvia}[1]{{I}^{\sta}_{#1}}
\newcommand{\rvib}[1]{{I}^{\stb}_{#1}}
\newcommand{\rvea}[1]{{F}^{\sta}_{#1}}
\newcommand{\rveb}[1]{{F}^{\stb}_{#1}}
\newcommand{\rvnia}[1]{{N}^{\sta}_{#1}}
\newcommand{\rvnib}[1]{{N}^{\stb}_{#1}}
\newcommand{\rvn}[1]{N_{#1}}
\newcommand{\rvr}[1]{D_{#1}}
\newcommand{\rvra}[1]{\typa{D}_{#1}}

\newcommand{\rvc}[1]{C_{#1}}
\newcommand{\rvca}[1]{\typa{C}_{#1}}
\newcommand{\rvcb}[1]{\typb{C}_{#1}}

\newcommand{\rva}{\typa{V}}

\newcommand{\stfreqa}{\alpha_{\sta}}
\newcommand{\stfreqb}{\alpha_{\stb}}

\newcommand{\exst}[0]{\mathbf{a}}
\newcommand{\frst}[0]{\mathbf{f}}
\newcommand{\curWin}[1]{\mathbf{d}(#1)}
\newcommand{\posit}[1]{#1^{{+}}}
\newcommand{\expstate}[1]{E(#1)}

\newcommand{\setMinus}[2]{#1\setminus#2}
\title{Optimal pattern matching algorithms}
\author{Gilles Didier\\
\small Aix-Marseille Universit\'e, CNRS, Centrale Marseille, I2M UMR7373, Marseille, France\\ 
\small E-mail: \url{gilles.didier@univ-amu.fr}}
\begin{document}
\maketitle
\begin{abstract}
We study a class of finite state machines, called \defi{$w$-matching machines}, which yield to simulate the behavior of pattern matching algorithms while searching for a pattern $w$.
They can be used to compute the asymptotic speed, i.e. the limit of the expected ratio of the number of text accesses to the length of the text, of algorithms while parsing an iid text to find the pattern $w$.

Defining the order of a matching machine or of an algorithm as the maximum difference between the current and accessed positions during a search (standard algorithms are generally of order $|w|$), we show that being given a pattern $w$, an order $k$ and an iid model, there exists an optimal $w$-matching machine, i.e. with the greatest asymptotic speed under the model among all the machines of order $k$, of which the set of states belongs to a finite and enumerable set.

It shows that it is possible to determine: 1) the greatest asymptotic speed among a large class of algorithms, with regard to a pattern and an iid model, and 2) a $w$-matching machine, thus an algorithm, achieving this speed.
\end{abstract}
\section{Introduction}

The problem of pattern matching consists in reporting all, and only the occurrences of a (short) word, a \defi{pattern}, $w$ in a (long) word, a \defi{text}, $t$. This question dates back to the early days of computer science. Since then, dozens of algorithms have been, and are still proposed, to solve it \cite{Faro2013}. Pattern matching has a wide range of applications: text, signal and image processing, database searching, computer viruses detection, genetic sequences analysis etc. Moreover, as a classic algorithmic problem, it may serve to introduce new ideas and paradigms in this field.  
Though optimal algorithms, in the sense of worst case analysis, have been developed forty years ago \cite{Knuth1977}, there exists as yet no algorithm which is fully efficient in all the various situations encountered in practice: large or small alphabets, long or short patterns etc. (see \cite{Faro2013}). 

The worst case analysis does not say much about the general behavior of algorithm in practical situations. In particular, the Knuth-Morris-Pratt algorithm is not much faster than the naive one, and even slower in average than certain algorithms with quadratic worst case complexity.
A more accurate measure of the algorithm efficiency in real situations is the average complexity on random texts or, equivalently, the expected complexity under a probabilistic  model of text.
The question of average case analysis of pattern matching algorithms was raised since at least \cite{Knuth1977}, in which the complexity of  pattern matching algorithms is conveniently expressed in terms of number of text accesses. A seminal work shows that, under the assumption that both the symbols of the pattern $w$ and the text are independently drawn uniformly from a finite alphabet, the minimum expectation of text accesses needed to search $w$ in $t$ is $\mathbf{O}\left(\frac{\size{t}\log\size{w}}{\size{w}}\right)$ \cite{Yao1979}.
Since then, several works studied the average complexity of some pattern matching algorithms, mainly Boyer-Moore-Horspool and Knuth-Morris-Pratt \cite{Yao1979,Guibas1981,Barth1984,BaezaYates1992,Mahmoud1997,Regnier1998,Smythe2001,Tsai2006,Marschall2008,Marschall2010,Marschall2011,Marschall2012}.
Different tools have been used to carry out these analysis, notably generating functions and Markov chains. In particular, G. Barth used Markov chains to compare the Knuth-Morris-Pratt and the naive algorithms \cite{Barth1984}. More recently, T. Marschall and S. Rahmann provided a general framework based on the same underlying ideas, for performing  statistical analysis of pattern matching algorithms, notably for computing the exact distributions of the number of text accesses of several pattern matching algorithms on iid texts \cite{Marschall2008,Marschall2010,Marschall2011,Marschall2012}. 

Following the same ideas, we consider finite state machines, called a \defi{$w$-matching machines}, which yield to simulate the behavior of pattern matching algorithms while searching a given pattern $w$.
They are used for studying the asymptotic behavior of pattern matching algorithms, namely the limit expectation of the ratio of the text length to the number of text accesses performed by an algorithm for searching a given pattern $w$ in iid texts, which we call the \defi{asymptotic speed} of the algorithm with regard to $w$ and the iid model. We show that the sequence of states of a  $w$-matching machine parsed while searching in an iid text follows a Markov chain, which yields to compute their asymptotic speed.

We next focus our interest in optimal $w$-matching machines, i.e. those with the greatest asymptotic speed with regard to an iid model (and the pattern $w$). The order of a $w$-matching machine (or of an algorithm with regard to $w$) is defined as the maximum difference between the current and the accessed positions during a search. Most of the $w$-matching   machines corresponding to standard algorithms are of order $\size{w}$ (a few of them have order $\size{w}+1$). We prove that, being given a pattern $w$, an order $k$ and an iid model, there exists an optimal $w$-matching machine of order $k$ in which the set of states is in bijection with the set of partial functions from $\{0, \ldots, k\}$ to the alphabet. It makes it possible to compute the greatest speed which can be achieved under a large class of algorithms (including all the pre-existing algorithms), and a $w$-machine achieving this speed. This optimal matching machine can be seen as a \textit{de novo} pattern matching algorithm which is optimal with regard to the pattern and the model. Some of the methods presented here have been implemented in the companion paper \cite{DidierY}. The software is available at \url{https://github.com/gilles-didier/Matchines.git}. 

The rest of the paper is organized as follows. Section \ref{s:defi} gives some basic notations and definitions. In Section  \ref{s:matching}, we present the $w$-matching machines and some of their properties. Next, we present three standard probabilistic models of text, define the asymptotic speed of an algorithm and show that the sequence of internal states of a $w$-matching machine follows a Markov chain on iid texts (Section \ref{s:models}). Last, in Section \ref{s:states}, we show that any $w$-matching machine can be ``simplified'' into a no slower $w$-matching machine of order $k$ with a set of states in bijection with the set of partial functions from $\{0, \ldots, k\}$ to the alphabet.

\section{Definitions and notations}\label{s:defi}
An \defi{alphabet} is a finite set $\alp$ of elements called \defi{letters} or \defi{symbols}.

A \defi{word}, a \defi{text} or a \defi{pattern} on $\alp$ is a finite sequence of symbols of $\alp$.
We put $\size{v}$ for the length of a word $v$ and $\size{v}_{w}$ for the number of occurrences of the word $w$ in $v$. The cardinal of the set $\set{S}$ is also noted $\size{\set{S}}$. Words are indexed from $0$, i.e. $v = v_{0}v_{1}\ldots v_{\size{v}-1}$. We put $v_{[i,j]}$ for the subword of $v$ starting at the position $i$ and ending at the position $j$, i.e. $v_{[i,j]} = v_{i}v_{i+1}\ldots v_{j}$. The concatenate of two words $u$ and $v$ is the word $uv=u_{0}u_{1}\ldots u_{\size{u}-1}v_{0}v_{1}\ldots v_{\size{v}-1}$.

For any length $n\geq 0$, we put $\alp^{n}$ for the set of words of length $n$ on $\alp$ and $\alp^{\star}$, for the set of finite words on $\alp$, i.e. $\alp^{\star} = \bigcup_{n=0}^{\infty} \alp^{n}$.

Unless otherwise specified, all the texts and patterns considered below are on a fixed alphabet $\alp$.

\section{Matching machines}\label{s:matching}

Let $w$ be a pattern on an alphabet $\alp$. A \defi{$w$-matching machine} is $6$-uple $(\stateset, \init, \prematchset, \next, \trans, \shift)$ where

\begin{itemize}
\item $\stateset$ is a finite number of states,
\item $\init\in\stateset$ is the initial state,
\item $\prematchset\subset\stateset$ is the subset of pre-match states,
\item $\next:\stateset\to \mathbb{N}$ is the next-position-to-check function, which is such that for all $q\in\prematchset$, $\next(q)<\size{w}$,
\item $\trans:\stateset\times\alp\to \stateset$ is the transition state function,
\item $\shift:\stateset\times\alp\to \mathbb{N}$ is the shift function.
\end{itemize}

By convention, the set of states of a matching machine always contains a \defi{sink state} $\sink$, which is such that, for all symbols $x\in\alp$,  $\trans(\sink, x) = \sink$ and $\shift(\sink, x) = 0$.

The \defi{order} $\ordmatchine$ of a matching machine $\matchine=(\stateset, \init, \prematchset, \next, \trans, \shift)$ is its greatest next-position-to-check, i.e. $\ordmatchine = \max_{q\in\stateset}\{\next(q)\}$.

Remark that $(\stateset, \init, \prematchset, \trans)$ is a deterministic finite automaton. The $w$-matching machines carry the same information as the \defi{Deterministic Arithmetic Automata} defined in \cite{Marschall2010,Marschall2011}.

\subsection{Generic algorithm}\label{secGen}
Algorithm \ref{algogen}, which will be referred to as the \defi{generic algorithm}, takes a $w$-matching machine and a text $t$ as input and is expected to output all the occurrence positions of $w$ in $t$.

\begin{algorithm}[htp]
\SetKw{and}{and}
\SetKw{print}{print}
\SetKwInOut{Input}{input}
\SetKwInOut{Output}{output}
	\SetAlgoLined\DontPrintSemicolon
	\Input{a $w$-matching machine $(\stateset, \init, \prematchset, \next, \trans, \shift)$ and a text $t$}
	\Output{all the occurrence positions of $w$ in $t$ (hopefully)}
	\vskip 0.2cm	
		$(\cState, \cPos) \leftarrow (\init, 0)$\;
		\While{$\cPos\leq\size{t}-\size{w}$} {
			\If{$\cState\in\prematchset$ \and $t_{\cPos+\next(\cState)} = w_{\next(\cState)}$}{
				\print ``~occurrence at position $\cPos$~''\;
			}
			$(\cState, \cPos) \leftarrow (\trans(\cState, t_{\cPos+\next(\cState)}), \cPos+\shift(\cState, t_{\cPos+\next(\cState)}))$\;
		}
		\vskip 0.3cm
\caption{Generic algorithm}\label{algogen}
\end{algorithm}

We put $\curState{\matchine}{t}{i}$ (resp. $\curPos{\matchine}{t}{i}$) for the state $\cState$ (resp. for the position $\cPos$) at the beginning of the $i^{\mbox{\tiny th}}$ iteration of the generic algorithm on the input $(\matchine, t)$.  We put $\curShift{\matchine}{t}{i}$ for the shift at the end of the $i^{\mbox{\tiny th}}$ iteration, i.e. $\curShift{\matchine}{t}{i} = \shift(\curState{\matchine}{t}{i}, t_{\curPos{\matchine}{t}{i}+\next(\curState{\matchine}{t}{i})})$. By convention, the generic algorithm starts with the iteration $0$.

A $w$-matching machine $\matchine$ is \defi{redundant} if there exist a text $t$ and two indexes $i<j$ such that
\begin{dmath*}
\curPos{\matchine}{t}{j}+\next(\curState{\matchine}{t}{j}) = \curPos{\matchine}{t}{i}+\next(\curState{\matchine}{t}{i})+\sum_{k=i}^{j-1}\curShift{\matchine}{t}{k}.
\end{dmath*}
In plain text, a matching machine $\matchine$ is redundant if there exists a text $t$ for which a position is accessed more than once during an execution of the generic algorithm on the input $(\matchine,t)$.

A $w$-matching machine $\matchine$ is \defi{valid} if, for all texts $t$, the execution of the generic algorithm on the input $(\matchine, t)$ outputs all, and only the occurrence positions of $w$ in $t$.

\begin{remark}\label{remValid}
If a matching machine is valid then 
\begin{enumerate}
\item its order is greater than or equal to $\size{w}-1$,
\item there is no text $t$ such that for some $j>i$, we have $\curState{\matchine}{t}{i}=\curState{\matchine}{t}{j}$ and $\shift(\curState{\matchine}{t}{k}, t_{\curPos{\matchine}{t}{k}+\next(\curState{\matchine}{t}{k})})=0$ for all $i\leq k \leq j$. In particular, the sink state is never reached during an execution of the generic algorithm with a valid machine.
\end{enumerate}
\end{remark}
\begin{proof}
Condition 1 comes from the fact that it is necessary  to check the position $(i+\size{w}-1)$  of the text to make sure whether $w$ occurs at $i$ or not. If Condition 2 is not fulfilled, an infinite loop starts at the $i^{\mbox{\tiny th}}$ iteration  of the generic algorithm on the input $(\matchine, t)$. In particular, the last occurrence of the pattern $w$ in the text $tw$ will never be reported.
\end{proof}

A \defi{match transition} is a transition going from a state $\cState\in\prematchset$ to the state $\trans(\cState,  w_{\next(\cState)})$. It corresponds to an actual match if the machine is valid.

A $w$-matching machine $\matchine$ is \defi{equivalent} to a $w$-matching machine $\matchine'$ if, for all texts $t$, the text accesses performed by the generic algorithm on the input $(\matchine, t)$ are the same as those performed on the input $(\matchine', t)$.
The machine $\matchine$ is \defi{faster} than $\matchine'$ if, for all texts $t$, the number of iterations of the generic algorithm on the input $(\matchine, t)$ is  smaller than that on the input $(\matchine', t)$.

We claim that, for all pre-existing pattern matching algorithms and all patterns $w$, there exists a $w$-matching machine $\matchine$ which is such that the text accesses performed by the generic algorithm on the input $(\matchine, t)$ are the exact same as those performed by the pattern matching algorithm on the input $(w, t)$.
Without giving a formal proof, this holds for any algorithm such that:
\begin{enumerate}
\item The current position in the text is stored in an internal variable which never decreases during their execution.
\item All the other internal variables, which will be refer to as \defi{state variables}, are bounded independently of the texts in which the pattern is searched.
\item The difference between the position accessed and the current position only depends on the state variables.
\end{enumerate}
We didn't find a pattern matching algorithm which not satisfies the conditions above.

Being given a pattern $w$, let us consider the $w$-matching machine where the set of states is made of the combinations of the possible values of the state variables, which are in finite number from Feature 2. Feature 3 ensures that we can define a next-position-to-check from the states of the machine, which is bounded independently from the input text. Last, the only changes which may occur between two text accesses during an execution of the algorithm, are an increment of the current position (Feature 1) and/or a certain number of modifications of the state variables, which ends up to change the state of the $w$-matching machine.
For instance the $w$-matching machine $\matchine$ associated to the naive algorithm has $\size{w}$ states  with 
\begin{itemize}
\item $\stateset = \{q_{0}, \ldots, q_{\size{w}-1}\}$,
\item $\init = q_{0}$, 
\item $\prematchset = \{q_{\size{w}-1}\}$,
\item $\next(q_{i}) = i$ for all indexes $0\leq i<w$,
\item $\trans(q_{i}, a) = \left\{\begin{array}{ll} q_{i+1}& \mbox{if $i<\size{w}-1$ and $a=w_{i}$,}\\
q_{0} & \mbox{otherwise,}\end{array}\right.$
\item $\shift(q_{i}, a) = \left\{\begin{array}{ll}0& \mbox{if $i<\size{w}-1$ and $a=w_{i}$,}\\
1& \mbox{otherwise.}\end{array}\right.$
\end{itemize}

A state $q$ of the matching machine $\matchine$ is \defi{reachable} in $\matchine$ if there exists a text $t$  such that $q$ is the current state of an iteration of the generic algorithm on the input $(\matchine, t)$. Unless otherwise specified or for temporary constructions, we will only consider matching machines $\matchine$ in which  all the states but the sink are reachable. Below, stating ``removing all the unreachable states'' will have to be understood as ``removing all the unreachable states but the sink''. Remark that all reachable states $q$ of a valid $w$-matching machine are such that there exists a text $t$ and two indexes $i\leq j$ such that $\curState{\matchine}{t}{i} = q$ and $\curState{\matchine}{t}{j} \in\prematchset$. In the same way,  a transition between two given states is reachable if there exists a text $t$ for which the transition occurs during the execution of the generic algorithm on the input $(\matchine, t)$.

We assume that for all pre-match states $q$ of $\matchine$, there exists a text $t$ such that a match transition starting from $q$ occurs during the execution of the generic algorithm on the input $(\matchine, t)$.

\subsection{Full-memory expansion -- standard matching machines}\label{secFull}
For all positive integers $n$, $\readset{n}$ denotes the set of subsets $H$ of  $\{0,\ldots, n\}\times\alp$ such that, for all $i\in\{0, \ldots, n\}$, there exists at most one pair in $H$ with $i$ as first entry. In other words, $\readset{n}$ is the set of partial functions from $\{0, \ldots, n\}$ to $\alp$.

For $H\in\readset{n}$, we put $\frst(H)$ for the set consisting of the first entries (i.e. the \defi{position entries}) of the pairs in $H$, namely  
\begin{dmath*}\frst(H)=\{i \hiderel{\condi}\exists x\hiderel{\in}\alp\mbox{ with }(i,x)\hiderel{\in} H\}.\end{dmath*} 

Let $k$ be a non-negative integer and $H\in\readset{n}$, the \defi{$k$-shifted} of $H$ is defined by 
\begin{dmath*}
\kshifted{k}{H} \hiderel{=} \{(u-k,y) \hiderel{\condi}  (u, y)\hiderel{\in} H \mbox{ with } u\hiderel{\geq} k\}.
\end{dmath*}
In plain text, $\kshifted{k}{H}$ is obtained by subtracting $k$ from the position entries of the pairs in $H$ and by keeping only the pairs with non-negative positive entries.

The \defi{full memory expansion} of a $w$-matching machine $\matchine = (\stateset, \init, \prematchset, \next, \trans, \shift)$ is the $w$-matching machine $\expan{\matchine}$ obtained by removing the unreachable states of the $w$-matching machine $\matchine' = (\stateset', \init', \prematchset', \next', \trans', \shift')$, defined as:
\begin{itemize}
\item $\stateset' = \stateset\times\readset{\ordmatchine}$,
\item $\init' = (\init, \emptyset)$,
\item $\next'((q,H)) = \next(q)$,
\item $\shift'((q,H), x) = \shift(q, x)$,
\item $\prematchset' = \prematchset\times\readset{\ordmatchine}$,
\item \begin{dmath*}\trans'((q,H), x) = \left\{\begin{array}{ll}
(\trans(q,x),  \kshifted{\shift(q, x)}{H\cup\{(\next(q),x)\}}) & \mbox{if $\next(q)\not\in\frst(H)$,}\\
\sink & \mbox{if $\exists a\neq x$ s.t. $(\next(q),a)\in H$,}\\
(\trans(q,x),  \kshifted{\shift(q, x)}{H}) & \mbox{if $(\next(q),x)\in H$.}
\end{array}\right.\end{dmath*}
\end{itemize}

\begin{remark}\label{remPos}
At the beginning of the $i^{\mbox{\tiny th}}$ iteration of the generic algorithm on the input $(\expan{\matchine}, t)$, if the current state is $(q, H)$ then the positions of $\{(j+\curPos{\expan{\matchine}}{t}{i}) \condi j\in\frst(H)\}$ are exactly the positions of $t$ greater than  $\curPos{\expan{\matchine}}{t}{i}$ which were accessed so far, while the second entries of the corresponding elements of $H$ give the symbols read.
\end{remark}

\begin{proposition}\label{propFull}
The $w$-matching machines $\matchine$ and $\expan{\matchine}$ are equivalent. In particular $\matchine$  is valid (resp. non-redundant) if and only if $\expan{\matchine}$ is valid (resp. non-redundant).
\end{proposition}
\begin{proof}
It is straightforward to prove by induction that, for all iterations $i$, if $\curState{\expan{\matchine}}{t}{i} = (q,H)$ then $\curState{\matchine}{t}{i} = q$, reciprocally, there exists $H\in\readset{\ordmatchine}$ such that $\curState{\expan{\matchine}}{t}{i} = (\curState{\matchine}{t}{i},H)$, $\curShift{\expan{\matchine}}{t}{i} = \curShift{\matchine}{t}{i}$ and $\curPos{\expan{\matchine}}{t}{i} = \curPos{\matchine}{t}{i}$.
\end{proof}

A $w$-matching machine $\matchine$ is \defi{standard} if its set of states has the same cardinal as that of its full memory expansion or, equivalently, if each state $q$ of $\matchine$ appears in a unique pair/state of its full memory expansion. For all states $q$ of a standard matching machine $\matchine$, we put $\mem{q}$ for the second entry of the unique pair/state of $\expan{\matchine}$ in which $q$ appears.

\begin{remark}\label{remStand}
Let $\matchine=(\stateset, \init, \prematchset, \next, \trans, \shift)$ be a standard matching machine. For all paths $q_{0}, \ldots, q_{n}$ of states in $\stateset_{\setminus{\{\sink\}}}$ in the DFA $(\stateset, \init, \stateset, \trans)$, there exists a text $t$ such that
$\curState{\matchine}{t}{i} = q_{i}$ for all $0\leq i \leq n$.
\end{remark}

\begin{theorem}\label{theoStdValid}
A standard and non-redundant $w$-matching machine $\matchine=(\stateset, \init, \prematchset, \next, \trans, \shift)$ is valid if and only if, for all $q\in\stateset$,
\begin{enumerate}
\item $q\in\prematchset$ if and only if we have $(j,w_{j})\in\mem{q}$ for all $j\in\setMinus{\{0,\ldots,\size{w}-1\}}{\{\next(q)\}}$,
\item \begin{dmath*}\shift(q,x)\leq\left\{\begin{array}{ll}\min\{k\geq 1 \condi w_{i-k} = y \mbox{ for all } (i,y)\in\mem{q}\cup\{(\next(q), x)\} \mbox{ with } k\leq i <k+\size{w}\} &\mbox{if } q\in\prematchset,\\
\min\{k\geq 0 \condi w_{i-k} = y \mbox{ for all } (i,y)\in\mem{q}\cup\{(\next(q), x)\} \mbox{ with } k\leq i <k+\size{w}\} &\mbox{otherwise,}
\end{array}\right.\end{dmath*}
\item there is no path $(q_{0},\ldots,q_{\ell})$ such that
\begin{itemize}
\item $q_{i}\neq q_{j}$ for all $0\leq i\neq j\leq \ell$,
\item there exists a word $v$ such that
\begin{itemize}
\item $q_{i+1} = \trans(q_{i}, v_{i})$ and $\shift(q_{i}, v_{i}) = 0$ for all $0\leq i <\ell$,
\item $q_{0} = \trans(q_{\ell}, v_{\ell})$ and $\shift(q_{\ell}, v_{\ell}) = 0$.
\end{itemize}
\end{itemize}
\end{enumerate}
\end{theorem}
\begin{proof}
We recall our implicit assumption that all the states of $\stateset$ are reachable. 
Let us assume that the property 1 of the theorem is not granted. Either there exists a state $q\in\prematchset$ and a position $j\in\setMinus{\{0,\ldots,\size{w}-1\}}{\{\next(q)\}}$ such that $(j,w_{j})\not\in\mem{q}$ or there exists a state $q\not\in\prematchset$ with  $(j,w_{j})\in\mem{q}$ for all $j\in\setMinus{\{0,\ldots,\size{w}-1\}}{\{\next(q)\}}$. From the implicit assumption, there exists a text $t$ and an iteration $i$ such that  $\curState{\matchine}{t}{i} = q$. Since $\matchine$ is non-redundant, the position $\curPos{\matchine}{t}{i}+\next(q)$ was not accessed before the iteration $i$ and we can assume that $t_{\curPos{\matchine}{t}{i}+\next(q)} = w_{\next(q)}$.
If $q\in\prematchset$ then the generic algorithm reports an occurrence of $w$ at $\curPos{\matchine}{t}{i}$. Furthermore, since $\matchine$ is standard, if there exists $j\in\setMinus{\{0,\ldots,\size{w}-1\}}{\{\next(q)\}}$ such that $(j,w_{j})\not\in\mem{q}$, then either the position $\curPos{\matchine}{t}{i}+j$ was accessed with $t_{\curPos{\matchine}{t}{i}+j}\neq w_{j}$ or it was not accessed and we can choose $t_{\curPos{\matchine}{t}{i}+j}\neq w_{j}$. In both cases, $w$ does not occur at $\curPos{\matchine}{t}{i}$ thus $\matchine$ is not valid. Let now assume that $q\not\in\prematchset$ and $(j,w_{j})\in\mem{q}$ for all $j\in\setMinus{\{0,\ldots,\size{w}-1\}}{\{\next(q)\}}$. This implies that $w$ does occur at the position $\curPos{\matchine}{t}{i}$ which is not reported at the iteration $i$. Since, from the definition of $w$-matching machines, the states $q'$ of $\prematchset$ are such that $\next(q')<\size{w}$ and $\matchine$ is non-redundant, the states parsed at iterations $j>i$ and such that $\curPos{\matchine}{t}{j}=\curPos{\matchine}{t}{i}$ are not in $\prematchset$. It follows that the position $\curPos{\matchine}{t}{i}$ is not reported at any further iteration. Again, $\matchine$ is not valid.

If the property 2 is not granted, it is straightforward to build a text $t$ for which an occurrence position of $w$ is not reported.

From the second item of Remark \ref{remValid}, if the property 3 is not granted then $\matchine$ is not valid.

Reciprocally, if $\matchine$ is not valid, there exist a text $t$ and a position $m$ for whose one of the following assertions holds:
\begin{enumerate}
\item the pattern $w$ occurs at the position $m$ of $t$ and $m$ is not reported by the generic algorithm on the input $(\matchine, t)$,
\item the generic algorithm reports the position $m$ on the input $(\matchine, t)$ but the pattern $w$ does not occurs at $m$.
\end{enumerate}
Let us first assume that the generic algorithm is such that $\curPos{\matchine}{t}{i}<m$ for all iterations $i$. Considering an iteration $i>(m+1)(\size{\stateset}+1)$, there exists an iteration $k\leq i$ such that for all $0\leq \ell\leq \size{\stateset}+1$, we have $\curShift{\matchine}{t}{k+\ell} = 0$, i.e. there exists a path $(q_{0},\ldots,q_{\ell})$ negating the property 3.

Let us now assume that there exists an iteration $i$ with $\curPos{\matchine}{t}{i}>m$ and let $k$ be the greatest index such that $\curPos{\matchine}{t}{k}\leq m$.
If $w$ occurs at the position $m$ which is not reported then the fact that  $\curPos{\matchine}{t}{k}\leq m$ and $\curPos{\matchine}{t}{k+1}>m$ contradicts the property 2.
Let us assume that $w$ does not occur at $m$ which is reported during the execution. We have necessarily $\curPos{\matchine}{t}{k}=m$, $\curState{\matchine}{t}{k}\in\prematchset$ and $\next(\curState{\matchine}{t}{k})<\size{w}$. If $w$ does not occur at $m$, then there exists a position $j\in\setMinus{\{0,\ldots,\size{w}-1\}}{\{\next(q)\}}$ such that 
$t_{\curPos{\matchine}{t}{k}+j}\neq w_{j}$ and, from Remark \ref{remPos}, we get $(j,w_{j})\not\in\mem{\curState{\matchine}{t}{k}}$ thus a contradiction with the property 1.
\end{proof}

Let $\matchine = (\stateset, \init, \prematchset, \next, \trans, \shift)$ be a matching machine and $\sta$ and $\stb$ be two states of $\stateset$. The \defi{redirected matching machine} $\redir{\matchine}{\sta}{\stb}$ is constructed from $\matchine$ by redirecting all the transitions that end with $\stb$, to $\sta$. Namely, the matching machine $\redir{\matchine}{\sta}{\stb}$ is obtained by removing the unreachable states of $\matchine' = (\stateset', \init', \prematchset', \next', \trans', \shift')$, defined for all $q\in\stateset_{\setminus \{\stb\}}$ and all symbols $x$, as:
\begin{itemize}
\item $\stateset' = \stateset_{\setminus \{\stb\}}$,
\item $\init' = \left\{\begin{array}{ll} \sta & \mbox{if } \stb=\init,\\ \init &\mbox{otherwise,}\end{array}\right.$
\item $\prematchset' =  \left\{\begin{array}{ll}
\prematchset_{\setminus \{\stb\}}\cup\{\sta\} & \mbox{if }  \stb\in\prematchset,\\
\prematchset&\mbox{otherwise,}
\end{array}\right.$
\item $\next'(q) = \next(q)$,
\item $\shift' = \shift(q, x)$,
\item $\trans'(q, x) = \left\{\begin{array}{ll}
\sta & \mbox{if }  \trans(q,x) = \stb,\\
\trans(q,x)&\mbox{otherwise.}
\end{array}\right.$
\end{itemize}

\begin{lemma}\label{lemmaRedir}
Let $\matchine$ be a standard $w$-matching machine and $\sta$ and $\stb$ be two states of $\stateset$ such that $\mem{\sta} = \mem{\stb}$. 
The redirected machines $\redir{\matchine}{\sta}{\stb}$ and $\redir{\matchine}{\stb}{\sta}$ are both standard. 
Moreover, if $\matchine$ is valid then both $\redir{\matchine}{\sta}{\stb}$  and $\redir{\matchine}{\stb}{\sta}$ are valid.
\end{lemma}
\begin{proof}
The fact that $\redir{\matchine}{\sta}{\stb}$ and $\redir{\matchine}{\stb}{\sta}$ are standard  comes straightforwardly from the fact that $\mem{\sta} = \mem{\stb}$. 

Let us assume that $\redir{\matchine}{\sta}{\stb}$ is not valid. There exist a text $t$ and a position $m$ for whose one of the following assertions holds:
\begin{enumerate}
\item the pattern $w$ occurs at the position $m$ of $t$ and $m$ is not reported by the generic algorithm on the input $(\redir{\matchine}{\sta}{\stb}, t)$,
\item the generic algorithm reports the position $m$ on the input $(\redir{\matchine}{\sta}{\stb}, t)$ but the pattern $w$ does not occurs at $m$.
\end{enumerate}

By construction, the smallest index $k$ such that $\curState{\redir{\matchine}{\sta}{\stb}}{t}{k}\neq \curState{\matchine}{t}{k}$  verifies $\curState{\redir{\matchine}{\sta}{\stb}}{t}{k}=\sta$, $\curState{\matchine}{t}{k}=\stb$ and $\curPos{\redir{\matchine}{\sta}{\stb}}{t}{i} = \curPos{\matchine}{t}{i}$ for all $i\leq k$.
If there is no iteration $j$ such that both $\curState{\redir{\matchine}{\sta}{\stb}}{t}{j} = \sta$ and $\curPos{\redir{\matchine}{\sta}{\stb}}{t}{j}\leq m$ then the executions of the standard algorithm coincide beyond the position $m$ on the inputs $(\matchine, t)$ and $(\redir{\matchine}{\sta}{\stb}, t)$. If $\redir{\matchine}{\sta}{\stb}$ is not valid then $\matchine$ is not valid.

Let us now assume that $q$ is reached before parsing the position $m$ on the input $(\redir{\matchine}{\sta}{\stb}, t)$ and let $j$ be the greatest index such that $\curState{\redir{\matchine}{\sta}{\stb}}{t}{j}=\sta$ and $\curPos{\redir{\matchine}{\sta}{\stb}}{t}{j}\leq m$.
Since the state $\sta$ is reachable with $\matchine$, there exists a text $u$ and an index $i$ such that $\sta$ is the current state of the $i^{\mbox{\tiny th}}$ iteration of the standard algorithm on the input $(\matchine, u)$. Let now define $v = u_{[0,\curPos{\matchine}{t}{i}-1]}t_{[\curPos{\redir{\matchine}{\sta}{\stb}}{t}{j}, \size{t}-1]}$. 
Since $\matchine$ is standard, the positions greater than $\curPos{\matchine}{u}{i}$ accessed by the generic algorithm on the input $(\matchine, u)$ at the $i^{\mbox{\tiny th}}$ iteration are  $\{(k+\curPos{\expan{\matchine}}{t}{i}) \condi k\in\frst(\mem{\sta})\}$ and the positions  greater than $\curPos{\redir{\matchine}{\sta}{\stb}}{t}{j}$ accessed by the generic algorithm at the $j^{\mbox{\tiny th}}$ on the input $(\redir{\matchine}{\sta}{\stb}, t)$ iteration are  $\{(k+\curPos{\redir{\matchine}{\sta}{\stb}}{t}{j} \condi k\in\frst(\mem{\sta})\}$ (Remark \ref{remPos}). When considered relatively to the current positions $\curPos{\expan{\matchine}}{t}{i}$ and $\curPos{\redir{\matchine}{\sta}{\stb}}{t}{j}$, the accessed positions greater than these current positions are the same. The positions accessed until the $i^{\mbox{\tiny th}}$ iteration on the inputs  $(\matchine, u)$ and $(\matchine, v)$, coincide. In particular, we have $\curState{\matchine}{v}{i}=\curState{\matchine}{u}{i}=\sta$. With the definitions of the text $v$, the execution of the generic algorithm from the $i^{\mbox{\tiny th}}$ on the input $(\matchine, v)$ does coincide with the execution of from the $j^{\mbox{\tiny th}}$ iteration on the input $(\redir{\matchine}{\sta}{\stb}, t)$. Again, if $\redir{\matchine}{\sta}{\stb}$ is not valid then $\matchine$ is not valid.
\end{proof}

\begin{lemma}\label{lemmaStandard}
Let $\matchine$ be a $w$-matching machine which is both valid and standard. For all states $q$ and all symbols $x$ and $y$, if $\trans(q, x) = \trans(q, y) \neq \sink$  then $\shift(q, x) = \shift(q, y)$.
\end{lemma}
\begin{proof}
Let us first remark that, since $\matchine$ is standard, the fact that $\trans(q, x) = \trans(q, y) \neq \sink$ implies that $\mem{\trans(q, x)} = \mem{\trans(q, y)}$. It follows that both $\shift(q, x)$ and $\shift(q, y)$ are strictly greater than $\next(q)$.

Let $d$ be the greatest position entry of the elements of $\mem{q}\cup\{(\next(q),x)\}$, i.e. $d = \max(\frst(\mem{q})\cup\{\next(q)\})$. By construction if $\mem{\trans(q, x)}$ (resp. $\mem{\trans(q, y)}$) is not empty then the greatest position entry of its elements is $d-\shift(q, x)$ (resp. $d-\shift(q, y))$. It follows that $\mem{\trans(q, x)} = \mem{\trans(q, y)}$ and $\shift(q, x)\neq\shift(q, y)$ is only possible if $\mem{\trans(q, x)} = \mem{\trans(q, y)}=\emptyset$, which implies that both $\shift(q, x)$ and $\shift(q, y)$ are strictly greater than $d$. 

Let us assume that $\shift(q, x)<\shift(q, y)$. We then have that $\shift(q, y)>d+1$. Let $t$ be a text such that there is a position $i$ with $\curState{\matchine}{t}{i} = q$, $t_{\curPos{\matchine}{t}{i}+\next(q)} = y$ and $w$ occurs at the position $\curPos{\matchine}{t}{i}+d$. Such a text $t$ exists since the state $q$ is reachable (with our implicit assumption) and the only positions of $t$ that we set, are not accessed until iteration $i$. Since $\shift(q, y)>d+1$, the occurrence of $w$ at the position $\curPos{\matchine}{t}{i}+d$ cannot be reported, which contradicts the assumption that $\matchine$ is valid.
\end{proof}

\subsection{Compact matching machines}
A $w$-matching machine $\matchine$ is \defi{compact} if it does not contain a state $\sta$ such that one of the following assertions holds:
\begin{enumerate}
	\item there exists a symbol $x$ with $\trans(\sta,x)\neq\sink$ and $\trans(\sta,y)=\sink$ for all symbols $y\neq x$;
	\item for all symbols $x$ and $y$, we have both $\trans(\sta,x) = \trans(\sta,y)$ and $\shift(\sta,x) = \shift(\sta,y)$.
\end{enumerate}

Let $\matchine = (\stateset, \init, \prematchset, \next, \trans, \shift)$ be a non-compact $w$-matching machine, $\sta$ be a state verifying one of the two assertions making $\matchine$ non-compact. 
If $\sta$ verifies the assertion 1 and  $x$ is the only symbol such that $\trans(\sta,x)\neq\sink$, we set $\trans(\sta,.)=\trans(\sta,x)$ and $\shift(\sta,.) = \shift(\sta,x)$. If  $\sta$ verifies the assertion 2, we set $\trans(\sta,.)=\trans(\sta,x)$ and $\shift(\sta,.) = \shift(\sta,x)$, by picking any symbol $x$.
The $w$-matching machine $\remo{\matchine}{\sta} = (\remo{\stateset}{\sta}, \remo{\init}{\sta}, \remo{\prematchset}{\sta}, \remo{\next}{\sta}, \remo{\trans}{\sta}, \remo{\shift}{\sta})$ is defined, for all states $q\in\stateset_{\setminus \{\sta\}}$ and all symbols $x$, as
\begin{itemize}
	\item $\remo{\stateset}{\sta} = \stateset_{\setminus \{\sta\}},$
	\item $\remo{\init}{\sta} = \left\{\begin{array}{ll} \init & \mbox{if }\sta\neq \init,\\
	\trans(\sta,.) & \mbox{otherwise,}\end{array}\right.$
	\item $\remo{\prematchset}{\sta} = \left\{\begin{array}{ll} \prematchset & \mbox{if }\sta\not\in\prematchset,\\ 
	\prematchset_{\setminus\{\sta\}} \cup \{q \condi\exists x\in\alp\mbox{ with } \trans(q,x) = \sta\} & \mbox{otherwise,}\end{array}\right.$
	\item $\remo{\next}{\sta}(q) = \next(q)$,\\
	\item $\remo{\trans}{\sta}(q,x) = \left\{\begin{array}{ll}
	 \trans(q,x) & \mbox{if }\trans(q,x)\neq \sta,\\
	 \trans(\sta,.) & \mbox{otherwise,}\end{array}\right.$
	\item $\remo{\shift}{\sta}(q,x) = \left\{\begin{array}{ll}
	 \shift(q,x) & \mbox{if }\trans(q,x)\neq \sta,\\
	 \shift(q,x) + \shift(\sta,.) & \mbox{otherwise.}\end{array}\right.${}
\end{itemize}
If all the states of $\matchine$ are reachable, then so are all the states of $\remo{\matchine}{\sta}$.

The following lemma ensures that any standard machine can be made compact and that this operation cannot deteriorate its efficiency.

\begin{lemma}\label{lemmaComp}
Let $\matchine$ be a $w$-matching machine which is made non-compact by a state $\sta$.{}
\begin{enumerate}
	\item If $\matchine$ is standard then $\remo{\matchine}{\sta}$ is standard.
	\item If $\matchine$ is valid then $\remo{\matchine}{\sta}$ is valid.
	\item $\remo{\matchine}{\sta}$ is faster than $\matchine$.
\end{enumerate}
\end{lemma}
\begin{proof}
We start by noting that if $\matchine$ is both standard and non-compact then there exist a state $\sta$ and a symbol $x$ such that $\trans(\sta,x)\neq\sink$ and $\trans(\sta,y)=\sink$ for all symbols $y\neq x$ (the other property leading to the non-compactness is excluded if $\matchine$ is standard). It follows that we have $(\next(\sta), x)\in\mem{\sta}$.  Redirecting all the transitions that end with $\sta$, to $\trans(\sta,x)$ and incrementing the shifts accordingly does not change the set $\mem{\trans(\sta,x)}$, nor any set $\mem{q}$. The matching machine $\remo{\matchine}{\sta}$ is still standard.

Let now assume that $\matchine$ is valid. In particular, the transitions to the sink state are never encountered (Remark \ref{remValid}). By construction, the sequence of states parsed during an execution of the generic algorithm with $\remo{\matchine}{\sta}$, can be obtained by withdrawing all the positions in which $\sta$ occurs from the sequence observed with $\matchine$. The machine $\remo{\matchine}{\sta}$ is thus valid and faster than the initial one.
\end{proof}

\begin{remark}\label{remComp}
	If a $w$-matching machine  $\matchine$  is both standard and compact then it is not redundant. 
\end{remark}

\begin{proposition}\label{propComp}
If $\matchine$ is a valid $w$-matching machine then there exists a standard, compact and valid $w$-matching machine $\matchine'$ which is faster or equivalent to $\matchine$.
\end{proposition}
\begin{proof}
By construction and from Proposition \ref{propFull}, the full memory expansion of $\matchine$ is both standard, valid and equivalent to $\matchine$. Next, applying Lemma \ref{lemmaComp} as long as there exist a state $q$ and a symbol $x$ such that $\trans(q,x)\neq\sink$ and $\trans(q,y)=\sink$ for all symbols $y\neq x$, leads to a compact, standard and valid $w$-matching machine faster  or equivalent to $\matchine$.
\end{proof}

\section{Random text models and asymptotic speed}\label{s:models}

\subsection{Text models}
A text model on an alphabet $\alp$ defines a probability distribution on $\alp^{n}$ for all lengths $n$.
Two text models are said \defi{equivalent} if they define the same probability distributions on $\alp^{n}$ for all lengths $n$.

We present three embedded classes of random text models, namely independent identically distributed, a.k.a. Bernoulli, Markov and Hidden Markov models.

An \defi{independent identically distributed (iid) model} is fully specified by a probability distribution $\piid$ on the symbols  of the alphabet. It will be simply referred to as ``$\piid$''. Under the model $\piid$, the probability of a text $t$ is 
\begin{dmath*}
	\proba{\piid}(t) = \prod_{i=0}^{\size{t}-1} \piid(t_{i}).
\end{dmath*}

A \defi{Markov model} $M$ of order $n$ is a $2$-uple $(\Minit, \Mtrans)$, where $\Minit$ is a probability distribution on the words of length $n$ of the alphabet (the initial distribution) and $\Mtrans$ associates a pair made of a word $u$ of length $n$ and a symbol $x$ with the probability for $u$ to be followed by $x$ (the transition probability). Under a Markov model $M = (\Minit, \Mtrans)$ of order $n$, the probability of a text $t$ of length greater than $n$ is 
\begin{dmath*}
\proba{M}(t) = \Minit(t_{[0,n-1]}) \prod_{i=n}^{\size{t}-1} \Mtrans(t_{[i-n,i-1]}, t_{i}).
\end{dmath*}

The probability distributions of words of length smaller than $n$ are obtained by marginalizing the distribution $\Minit$. Under this definition, Markov models are homogeneous (i.e. such that the transition probabilities do not depend on the position). ``Markov model'' with no order specified stands for ``Markov model of order 1''.

A \defi{Hidden Markov model (HMM)} $H$ is a $4$-uple $(\Hstate, \Hinit, \Htrans, \Hemis)$  where  $\Hstate$ is a set of (hidden) states, $(\Hinit, \Htrans)$ is a Markov model of order $1$ on $\Hstate$, and $\Hemis$ associates a pair made of a state $q$ and of a symbol $x$ of the text alphabet with the probability for the state $q$ to emit $x$ (i.e. $\Hemis(q,.)$ is a probability distribution on the text alphabet). Under a HMM $H$, the probability of a text $t$ is 
\begin{dmath*}
\proba{H}(t) = \sum_{q\in\Hstate^{\size{t}}} \Hinit(q_{0})\Hemis(q_{0}, t_{0})\prod_{i=1}^{\size{t}-1} \Htrans(q_{i-1}, q_{i})\Hemis(q_{i}, t_{i}).
\end{dmath*}

We will often consider HMMs $H=(\Hstate, \Hinit, \Htrans, \Hemis)$ with \defi{deterministic emission functions}, i.e. such that for all states $d\in\Hstate$ there exists a unique symbol $x$ with $\Hemis(d,x)>0$, i.e. with $\Hemis(d,x)=1$. In this case, for all states $d$, we will put $\HemisBis(d)$ for the unique symbol such that $\Hemis(d,\HemisBis(d))>0$ ($\HemisBis$ is just a map from $\Hstate$ to the alphabet). Remark that for all HMM $H$, there exists a HMM $H'$ with a deterministic emission function which is equivalent to $H$ (it is obtained by splitting the hidden states according to the symbols emitted and by setting the probability transitions accordingly). 
In \cite{Marschall2010,Marschall2011}, authors define the \defi{finite-memory text models} which are essentially HMMs with an additional emission function.

Basically, iid models are special cases of Markov models which are themselves special cases of HMMs.

The next theorem is essentially a restatement of  Item 1 of Lemma 3 in \cite{Marschall2011}, for matching machines and HMMs.

\begin{theorem}[\cite{Marschall2011}]\label{theoHMM}
	Let $\matchine = (\stateset, \init, \prematchset, \trans, \next, \shift)$ be a $w$-matching machine. If a text $t$ follows an HMM then there exists a Markov model $(\Hinit[H'], \Htrans[H'])$  of state set $\Hstate[H']$ such that there exist:
	\begin{itemize}
		\item a map $\HemisBis[H']^{[t]}$ from $\Hstate[H']$ to $\alp$ such that $t$ follows the HMM with deterministic emission $(\Hstate[H'], \Hinit[H'], \Htrans[H'], \HemisBis[H']^{[t]})$,
		\item a map $\HemisBis[H']^{[\stateset]}$ from $\Hstate[H']$ to $\stateset$ such that $(\curState{\matchine}{t}{i})_{{i}}$ follows the HMM with deterministic emission  $(\Hstate[H'], \Hinit[H'], \Htrans[H'], \HemisBis[H']^{[\stateset]})$
		\item a map $\HemisBis[H']^{[s]}$ from $\Hstate[H']$ to $\{0,\ldots,\size{w}\}$ such that $(\curShift{\matchine}{t}{i})_{{i}}$ follows  the HMM with deterministic emission $(\Hstate[H'], \Hinit[H'], \Htrans[H'], \HemisBis[H']^{[s]})$
	\end{itemize}
	
%
\end{theorem}
\begin{proof}
We assume without loss of generality that $t$ follows an HMM $H$ with a deterministic emission, $H=(\Hstate, \Hinit, \Htrans, \HemisBis)$.

We set $\Hstate[H'] = \Hstate^{\ordmatchine+1}\times\stateset$. Let $(\Hinit[H'], \Htrans[H'])$ be the Markov model on $\Hstate[H']$ such that for all $d, d'\in\Hstate^{\ordmatchine+1}$ and all $q, q'\in\stateset$, we have

\begin{dmath*}
	\Hinit[H']([d, q]) = \left\{\begin{array}{ll} \proba{(\Hinit, \Htrans)}(d) & \mbox{if $q = \init$,}\\
0 & \mbox{otherwise, and}\end{array}\right.
\end{dmath*}
\begin{dmath*}	
\Htrans[H']([d, q], [d',q']) = \left\{\begin{array}{ll} 
0 & \mbox{if }q'\neq \trans(q, \HemisBis(d_{{\next(q)}})),\\
1 & \mbox{if }q'= \trans(q, \HemisBis(d_{\next(q)})), d'= d\mbox{ and } \shift(q, \HemisBis(d_{\next(q)})) = 0,\\
\probaS{(\Hinit, \Htrans)}(d,d',\shift(q, \HemisBis(d_{\next(q)}))) & \mbox{if } q'= \trans(q, \HemisBis(d_{\next(q)}))\mbox{ and } \shift(q, \HemisBis(d_{\next(q)}))>0,
 \end{array}\right.
 \end{dmath*}
 where $\probaS{(\Hinit, \Htrans)}(d,d',\ell)$ is the probability of observing $d'$ given that  $d$ occurs $\ell$ positions before, under the Markov model $(\Hinit, \Htrans)$,


Since the emission function of $H$ is deterministic, a sequence of hidden states $z$ of $\Hstate$ determines the emitted text $t^{z} = \HemisBis(z)$, which itself determines the sequence $\left(\curState{\matchine}{t^{z}}{i}, \curShift{\matchine}{t^{z}}{i}\right)_{i}$ of pairs state-shift parsed on the input $(\matchine, t^{z})$. Let us verify that if $z$ follows the Markov model $(\Hinit, \Htrans)$, then the Markov model $(\Hinit[H'], \Htrans[H'])$ models the sequence 
$\left([\curWin{i}, \curState{\matchine}{t^{z}}{i}]\right)_{i}$,
where $\curWin{i} = z_{[\curPos{\matchine}{t^{z}}{i},\curPos{\matchine}{t^{z}}{i}+\ordmatchine]}$.

Under the current assumptions and since the generic algorithm always starts with $\init$, we have $\curState{\matchine}{t^{z}}{0}=\init$ and $\p\left([\curWin{0}, \curState{\matchine}{t^{z}}{0}]\right) = \proba{(\Hinit, \Htrans)}(z_{[0,\ordmatchine]})$. The initial state of the sequence $\left([\curWin{i}, \curState{\matchine}{t^{z}}{i}]\right)_{i}$ does follow the distribution  $\Hinit[H']$. 

Let us assume that, for $j\geq 0$, the probability of  $\left([\curWin{i}, \curState{\matchine}{t^{z}}{i}]\right)_{[0, j]}$ is 
\begin{dmath*}
\proba{(\Hinit', \Htrans')}\left(([\curWin{i}, \curState{\matchine}{t^{z}}{i}])_{[0, j]}\right).
\end{dmath*}

Both the next state  $\curState{\matchine}{t^{z}}{j+1}$ and the shift $\curShift{\matchine}{t^{z}}{j}$ only depend on $\curState{\matchine}{t^{z}}{j}$ and on the symbol $x_{j} = t^{z}_{\curPos{\matchine}{t^{z}}{j}+\next(\curState{\matchine}{t^{z}}{j})}$. Both are fully determined by the current state $[\curWin{j}, \curState{\matchine}{t^{z}}{j}]$ of $\Hstate[H']$.
In particular, for all $d\in\Hstate^{\ordmatchine+1}$, if we have $q'\neq\trans(\curState{\matchine}{t^{z}}{j}, x_{j})$ then 
\begin{dmath*}
[\curWin{j+1}, \curState{\matchine}{t^{z}}{j+1}]\neq[d, q'].
 \end{dmath*}
 
If  $\curShift{\matchine}{t^{z}}{j} = 0$ then we have
\begin{dmath*}
[\curWin{j+1}, \curState{\matchine}{t^{z}}{j+1}]=[\curWin{j}, \trans(\curState{\matchine}{t^{z}}{j}, x_{j})]
\condition{with probability $1$.}
 \end{dmath*}

Otherwise, by setting
 $\curPos{\matchine}{t^{z}}{j+1} = \curPos{\matchine}{t^{z}}{j}+\curShift{\matchine}{t^{z}}{j}$, we get 
\begin{dmath*}
[\curWin{j+1}, \curState{\matchine}{t^{z}}{j+1}]=[\curWin{j+1}, \trans(\curState{\matchine}{t^{z}}{j}, x_{j})]\condition{with probability $\probaS{(\Hinit, \Htrans)}(\curWin{j},\curWin{j+1},\curShift{\matchine}{t^{z}}{j})$.}
\end{dmath*}

Altogether, we get that the probability of  $\left([\curWin{i}, \curState{\matchine}{t^{z}}{i}]\right)_{[0, j+1]}$ is equal to
\begin{dmath*}
\proba{(\Hinit', \Htrans')}\left(([\curWin{i}, \curState{\matchine}{t^{z}}{i}])_{[0, j+1]}\right).
\end{dmath*}

The sequence $\left([\curWin{i}, \curState{\matchine}{t^{z}}{i}]\right)_{i}$ follows the Markov model $(\Hinit[H'], \Htrans[H'])$. By construction, .
\end{proof}

Theorem \ref{theoHMM} holds for both Markov and iid models and implies that both  the sequence of state and the sequence of shifts follow an HMM. If $t$ follows a Markov model of order $n$, one can prove in the same way that the sequence $(t_{[k_{i}, k_{i}+L-1]},\curState{\matchine}{t}{i})_{i}$ with $L = \max\{\ordmatchine, n\}$, follows a Markov model, which may emit the sequence of states and that of shifts. More interestingly, if $t$ follows an iid model and $\matchine$ is non-redundant or standard then the sequence of states parsed on the input $(\matchine, t)$ directly follows a Markov model.

\begin{theorem}\label{theoiid}
	Let $\matchine = (\stateset, \init, \prematchset, \next, \trans, \shift)$ be a $w$-matching machine. If a text $t$ follows an iid model and $\matchine$ is non-redundant (resp. standard) then the sequence of states parsed by the generic algorithm on the input $(\matchine, t)$ follows a Markov model $M = (\Minit, \Mtrans)$, where for all states $q$ and $q'$,
	\begin{itemize}
		\item $\Minit(q)=\left\{\begin{array}{ll} 1 & \mbox{if $q = 0$,} \\ 0 & \mbox{otherwise;}\end{array}\right.$
		\item $\Mtrans(q,q') = \displaystyle\smashoperator[r]{\sum_{x,\trans(q,x) = q'}} \piid(x)$ if $\matchine$ is not redundant;
		\item $\Mtrans(q,q') = \frac{\displaystyle\smashoperator[r]{\sum_{{x, \trans(q,x) = q'}}} \piid(x)}{\displaystyle\smashoperator[r]{\sum_{{x,\trans(q,x) \neq \sink}}} \piid(x)}$ if $\matchine$ is standard.
	\end{itemize}
\end{theorem}
\begin{proof}
Whatever the text model and the matching machine, the sequence of states always starts with the state $\init$ with probability $1$. We have $\Minit(o)=1$ and $\Minit(q)=0$ for all $q\neq o$.

If the positions of $t$ are iid with distribution $\piid$ and if $\matchine$ is non-redundant then the symbols read at each text access are independently drawn from $\piid$. It follows that the probability that the state $q'$ follows the state $q$ at any iteration is 
\begin{dmath*}
\Mtrans(q,q') = \displaystyle\smashoperator[r]{\sum_{x,\trans(q,x) = q'}} \piid(x),\end{dmath*} independently of the previous states.

Let us now assume that $\matchine$ is standard and that the text $t$ still follows an iid model $\piid$. 
By construction, the probability $\Mtrans(q,q')$ that the state $q'$ follows the state $q$ during the execution of the generic algorithm on the input $(\matchine, t)$, is equal to:
\begin{itemize}
\item $1$, if there exists a symbol $x$ such that $(\next(q), x)\in\mem{q}$ and $\trans(q,x) = q'$,
\item $\displaystyle\smashoperator[r]{\sum_{x,\trans(q,x) = q'}} \piid(x)$, otherwise,
\end{itemize}
 independently of the previous states.
If there exists a symbol $x$ such that $(\next(q), x)\in\mem{q}$, the we have $\trans(q,y) = \sink$ for all symbols $y\neq x$. Otherwise, since $\matchine$ is valid, there is no symbol $y$ such that $\trans(q,y) = \sink$. In both cases, we have that 
\begin{dmath*}
\Mtrans(q,q') = \frac{\displaystyle\smashoperator[r]{\sum_{{x, \trans(q,x) = q'}}} \piid(x)}{\displaystyle\smashoperator[r]{\sum_{{x,\trans(q,x) \neq \sink}}} \piid(x)}
\end{dmath*}
\end{proof}

\subsection{Asymptotic speed}

In \cite{Marschall2010,Marschall2011}, the authors studied the exact distribution of the number of text accesses of some classical algorithms seeking for a pattern in Bernoulli random texts of a given length. We are here rather interested in the asymptotic behavior of algorithms, still in terms of text accesses.

Let $\model$ be a  text model and $\alg$ be an algorithm. The \defi{asymptotic speed} of $\alg$ with respect to $w$ and under $\model$ is the limit, when $n$ goes to infinity,  of the expectation of the ratio of $n$ to the number of text accesses performed by $\alg$ by parsing a text of length $n$ drawn from $\model$. Formally, by putting $\tac{\alg}(t)$ for the number of text accesses performed by $\alg$ to parse $t$, the asymptotic speed of $\alg$ under $\model$ is
\begin{dmath*}
\as{\model}{\alg} = \lim_{n\to\infty}\sum_{t\in\alp^{n}} \frac{\size{t}}{\tac{\alg}(t)}\prob_{\model}(t).
\end{dmath*}

In order to make the notations less cluttered, $w$ does not appear neither on $\as{\model}{\alg}$ nor on $\tac{\alg}(t)$, but these two quantities actually depend on $w$.
At this point, nothing ensures that the limit above exists.

For all $w$-matching machines $\matchine$, we put $\tac{\matchine}$ for the number of text accesses and $\as{\model}{\matchine}$ for the asymptotic speed of the generic algorithm with $\matchine$ as first input. For a matching machine, the number of text accesses coincides with the number of iterations.

The following remark is a direct consequence of the definition of redundancy and of Remark \ref{remComp}.

\begin{remark}
If it exists, the asymptotic speed of a non-redundant matching machine is greater than $1$.
\end{remark}

In particular, the remark above holds for $w$-matching machines which are both standard and compact (Remark \ref{remComp}). It implies that any matching machine can be turned into a matching machine with an asymptotic speed greater than $1$ (Proposition \ref{propComp}).

\begin{lemma}\label{lemmaShift}
Let $\matchine = (\stateset, \init, \prematchset, \next, \trans, \shift)$ be a $w$-matching machine. If $\matchine$ is valid then we have for all texts $t$,
\begin{dmath*}
\left\lfloor\frac{\size{t}}{\size{w}}\right\rfloor\hiderel{\leq}\tac{\matchine}(t)\hiderel{\leq}(\size{t}+1)(\size{\stateset}+1).
\end{dmath*}
\end{lemma}
\begin{proof}
If there exists a text $t$ such that $\tac{\matchine}(t)<\left\lfloor\frac{\size{t}}{\size{w}}\right\rfloor$ then there exists $\size{w}$ successive positions of $t$ which are not accessed during the execution of the generic algorithm on the input $(\matchine, t)$ \cite{Knuth1977}. They may contain an occurrence of $w$ which wouldn't be reported.

If there exists a text $t$ such that $\tac{\matchine}(t)>(\size{t}+1)(\size{\stateset}+1)$ then there exists an iteration $i\leq \tac{\matchine}(t)-\size{\stateset}-1$ such that $\curShift{\matchine}{t}{j}=0$ for all $i\leq j \leq i+\size{\stateset}$. Since there are only $\size{\stateset}$ states, there exist two integers $k$ and $\ell$ such that $i\leq k<\ell \leq i+\size{\stateset}$ and $\curState{\matchine}{t}{k}=\curState{\matchine}{t}{\ell}$, which contradicts the validity of $\matchine$ (Item 2 of Remark \ref{remValid}).
\end{proof}

We will need the following technical lemma.

\begin{lemma}\label{lemmaConv}
Let $M= (\Minit, \Mtrans)$ be a Markov model on an alphabet $\Mstate$ and $\increas$ be a map from $\Mstate$ to $\mathbb{N}$. 
Let us assume that we have
\begin{dmath*}
\lim_{n\to\infty} \sum_{i=0}^{n}\increas(V_{i}) = \infty\condition[]{with probability $1$,}
\end{dmath*}
where $(V_{i})_{i}$ is the Markov chain in which $V_{0}$ has the probability distribution $\Minit$ and, for all $i\geq 0$, $\p\{V_{i+1} = b  \condi V_{i} = a\} = \Mtrans(a,b)$.

By setting  $\subinc{n} = \{v\in\Mstate^{*} \condi\sum_{i=0}^{\size{v}-2}\increas(v_{i})+\kappa< n \leq \sum_{i=0}^{\size{v}-1}\increas(v_{i})+\kappa\}$ where $\kappa$ is a non-negative number, the sum
\begin{dmath*}
\sum_{v\in\subinc{n}} \frac{\size{v}_{x}}{\size{v}}\proba{M}(v)
\end{dmath*}
converges for all states $x\in\Mstate$ as $n$ goes to infinity, to
\begin{dmath*}
\lim_{k\to\infty}\sum_{v\in\Mstate^{k}} \frac{\size{v}_{x}}{k}\proba{M}(v).
\end{dmath*}
\end{lemma}
\begin{proof}

We define the random variable $F_{x,n}$ as the ratio $\frac{\size{V_{[0,\ell_{V,n}-1]}}_{x}}{\ell_{V,n}}$ where $\ell_{V,n}$ is the smallest integer such that $\sum_{i=0}^{\ell_{V,n}-1}\increas(V_{i})+\kappa\geq n$.

Since, under the assumptions of the lemma, $\lim_{n\to\infty}\ell_{V,n}=\infty$ with probability $1$, we have  
\begin{dmath}\label{eqLimit}
\lim_{n\to\infty} F_{x,n} =  \lim_{n\to\infty}\frac{\size{V_{[0,\ell_{V,n}-1]}}_{x}}{\ell_{V,n}} =\lim_{k\to\infty} \frac{\size{V_{[0,k-1]}}_{x}}{k} \condition[]{with probability $1$.}
\end{dmath}

The fact that $\frac{\size{V_{[0,k-1]}}_{x}}{k}$ converges almost surely (a.s.) as $k$ goes to $\infty$ is a classical result of Markov chains. In particular, The Ergodic Theorem states that if the chain is irreducible $\frac{\size{V_{[0,k-1]}}_{x}}{k}$ converges a.s. to the probability of the state $x$ in its stationary distribution \cite{Feller1968}.

Let us  remark that, for all $v\in\Mstate^{*}$, the probability $\p\left\{\ell_{V,n} = \size{v} \condi V_{[0, \size{v}-1]}=v\right\}$
 is $1$ if $v$ verifies $\sum_{i=0}^{\size{v}-2}\increas(v_{i})+\kappa< n \leq \sum_{i=0}^{\size{v}-1}\increas(v_{i})+\kappa$, and $0$ otherwise. We have, for all $v\in\Mstate^{*}$, 
\begin{dmath*}
\p\{\ell_{V,n} = \size{v} \mbox{ and }V_{[0, \size{v}-1]} = v\} = \left\{\begin{array}{ll}\proba{M}(v)& \mbox{if } v\in\subinc{n},\\ 0 &\mbox{otherwise.}\end{array}\right.
\end{dmath*}

It follows that
\begin{dmath*}
\expect(F_{x,n})=\sum_{k>0}\left(\sum_{v\in\Mstate^{k}} \frac{\size{v}_{x}}{k}\p\{\ell_{V,n} \hiderel{=} k \mbox{ and } V_{[0, k-1]} \hiderel{=} v\}\right)
= \sum_{v\in\subinc{n}} \frac{\size{v}_{x}}{\size{v}}\proba{M}(v).
\end{dmath*}

Moreover, since $F_{x,n}\leq 1$, the bounded convergence theorem gives us that
\begin{dmath*}
\lim_{n\to\infty} \expect(F_{x,n}) = \expect(\lim_{n\to\infty} F_{x,n}).
\end{dmath*}
The sum $\sum_{v\in\subinc{n}} \frac{\size{v}_{x}}{\size{v}}\proba{M}(v)$ does converge as $n$ goes to $\infty$, to 
\begin{dmath*}\lim_{k\to\infty}\sum_{v\in\Mstate^{k}} \frac{\size{v}_{x}}{k}\proba{M}(v)\condition[]{(Equation \ref{eqLimit}).}
\end{dmath*}
\end{proof}
We are now able to prove that the asymptotic speed of a matching machine does exist under an HMM.

\begin{theorem}\label{theoASExist}
Let $H$ be a HMM and $\matchine$ be a valid matching machine. The sum  $\sum_{t\in\alp^{n}} \frac{\size{t}}{\tac{\matchine}(t)}\prob_{H}(t)$ converges as $n$ goes to infinity.
\end{theorem}
\begin{proof}
Let $H=(\Hstate, \Hinit, \Htrans, \Hemis)$ be a HMM and $t$ be a text. The number of iterations of the generic algorithm on the input $(\matchine, t)$ is equal to the number $\tac{\matchine}(t)$ of text accesses. From the loop condition of the generic algorithm, we get that 

\begin{dmath}\label{eqIneq}
\frac{\sum_{i=0}^{\tac{\matchine}(t)-2}\curShift{\matchine}{t}{i}}{\tac{\matchine}(t)} + \frac{\size{w}}{\tac{\matchine}(t)}\hiderel{<}\frac{\size{t}}{\tac{\matchine}(t)} \hiderel{\leq} \frac{\sum_{i=0}^{\tac{\matchine}(t)-1}\curShift{\matchine}{t}{i}}{\tac{\matchine}(t)} + \frac{\size{w}}{\tac{\matchine}(t)}.
\end{dmath}

Since the validity of $\matchine$ implies that $\lim_{\size{t}\to\infty}\tac{\matchine}(t) = \infty$ (Lemma \ref{lemmaShift}), Inequality \ref{eqIneq} leads to
\begin{dmath}\label{eqLim}
\lim_{n\to\infty}\sum_{t\in\alp^{n}} \frac{\size{t}}{\tac{\matchine}(t)}\prob_{H}(t) = \lim_{n\to\infty}\sum_{t\in\alp^{n}} \frac{\sum_{i=0}^{\tac{\matchine}-1}\curShift{\matchine}{t}{i}}{\tac{\matchine}(t)}\prob_{H}(t).
\end{dmath}

From Theorem \ref{theoHMM}, if $t$ follows the HMM $H$ then the sequence of shifts $(\curShift{\matchine}{t}{i})_{0\leq i<\tac{\matchine}(t)}$ follows a HMM $H'=(\Hstate[H'], \Hinit[H'], \Htrans[H'], \HemisBis[H'])$, which is assumed to have a deterministic emission without loss of generality. $H'=(\Hstate[H'], \Hinit[H'], \Htrans[H'], \Hemis[H'])$. 
By setting  
\begin{dmath*}
\subinc{n} \hiderel{=} \left\{v\hiderel{\in}\Hstate[H']^{*} \condi\sum_{i=0}^{\size{v}-2}\HemisBis[H'](v_{i})+\kappa\hiderel{<} n \leq \sum_{i=0}^{\size{v}-1}\HemisBis[H'](v_{i})+\kappa\right\},
\end{dmath*}
Equation \ref{eqLim} becomes

\begin{dmath*}\label{eqBof}
\lim_{n\to\infty}\sum_{t\in\alp^{n}} \frac{\size{t}}{\tac{\matchine}(t)}\prob_{H}(t) = \lim_{n\to\infty}\sum_{v\in\subinc[\size{w}]{n}} \frac{\sum_{i=0}^{\size{v}-1}\HemisBis[H'](v_{i})}{\size{v}}\proba{H'}(v)\\
 = \lim_{n\to\infty}\sum_{v\in\subinc[\size{w}]{n}} \sum_{q\in\Hstate[H']}\HemisBis[H'](q)\frac{\size{v}_{q}}{\size{v}}\proba{H'}(v).
\end{dmath*}
Interchanging the order of summation gives us that
\begin{dmath}\label{eqExist}
\sum_{v\in\subinc[\size{w}]{n}} \sum_{d\in\Hstate[H']}\HemisBis[H'](d)\frac{\size{v}_{d}}{\size{v}}\proba{H'}(v) = \sum_{d\in\Hstate[H']}\HemisBis[H'](d)\left(\sum_{v\in\subinc[\size{w}]{n}} \frac{\size{v}_{d}}{\size{v}}\proba{H'}(v)\right).
\end{dmath}
Since the sequence of shifts follows $H'$ when the text follows $H$, for all $v\in\Hstate[H']^{*}$ such that $\proba{(\Hinit[H'], \Htrans[H'])}(v)\hiderel{>}0$, there exists a text $t$ with $\proba{H}(t)>0$ and such that the sequence of shifts parsed on the input $(\matchine, t)$ is $\HemisBis[H'](v)$ and $\size{v}$ is the number of iterations (or text accesses). Under the assumption that $\matchine$ is valid, Lemma \ref{lemmaShift} implies that 
\begin{dmath*}
\size{v}\leq(\size{\Hstate[H']}+1)(\size{t}+1)\leq(\size{\Hstate[H']}+1)\left(\sum_{i=0}^{\size{v}}\HemisBis[H'](v_{{i}})+1\right)
\end{dmath*}
thus
\begin{dmath*}
\sum_{i=0}^{|v|-1}\HemisBis[H'](v_{{i}})\geq \frac{\size{v}}{\size{\Hstate[H']}+1}-1
\end{dmath*}
In short, we have $\proba{(\Hinit[H'], \Htrans[H'])}(v)\hiderel{>}0\Rightarrow\sum_{i=0}^{|v|-1}\HemisBis[H'](v_{{i}})\geq \beta\size{v}$ with $\beta>0$, which implies that the Markov model $(\Hinit[H'], \Htrans[H'])$ and the map  $\HemisBis[H']$ satisfy the assumptions of
Lemma \ref{lemmaConv}. We get that the sum $\sum_{v\in\subinc[\size{w}]{n}} \frac{\size{v}_{q}}{\size{v}}\proba{H'}(v)$ converges to a limit frequency $\freq_{q}$ as $n$ goes to $\infty$.
From Equation \ref{eqExist}, the asymptotic speed $\as{H}{\matchine}$ does exist and is equal to
\begin{dmath*}
\sum_{q\in\Hstate[H']}\HemisBis[H'](q)\freq_{q}.
\end{dmath*}
\end{proof}

A more precise result can be stated in the case where the model is Bernoulli and the machine is standard.

\begin{theorem}
Let $\matchine = (\stateset, \init, \prematchset, \next, \trans, \shift)$ be a standard and valid $w$-matching machine and $\piid$ a Bernoulli model. The asymptotic speed of $\matchine$ under $\piid$ is equal to:
\begin{dmath*}\label{eqAsympFreq}
\as{\piid}{\matchine} = \sum_{q\in\stateset} \alpha_{q}\expstate{q},
\end{dmath*}
where $(\alpha_{q})_{q\in{\stateset}}$ are the limit frequencies of the states of the Markov model associated to $\matchine$ and $\piid$, given in Theorem \ref{theoiid} and 
\begin{dmath*}
\expstate{q} = \frac{\sum_{x,\trans(q,x) \neq \sink} \shift(q,x)\piid_{x}}{\sum_{x,\trans(q,x) \neq \sink} \piid_{x}}.
\end{dmath*}
\end{theorem}
\begin{proof}
Since $\matchine$ is standard and valid, Lemma \ref{lemmaStandard} states that 
any transition from a state $r$ to a state $s$ (whatever the symbol read from the text) is associated to a unique shift which will be referred to as $\myshift(r,s)$. 

For all texts $t$, we then have
\begin{dmath*}\label{eqIneqBis}
\frac{\sum_{i=0}^{\tac{\matchine}(t)-3}\myshift(\curState{\matchine}{t}{i},\curState{\matchine}{t}{i+1})}{\tac{\matchine}(t)} + \frac{\size{w}}{\tac{\matchine}(t)}\hiderel{<}\frac{\size{t}}{\tac{\matchine}(t)} \hiderel{\leq} \frac{\sum_{i=0}^{\tac{\matchine}(t)-2}\myshift(\curState{\matchine}{t}{i},\curState{\matchine}{t}{i+1})}{\tac{\matchine}(t)} + \frac{\size{w}}{\tac{\matchine}(t)}.
\end{dmath*}
Theorem \ref{theoiid} tells us that if $t$ is drawn according $\piid$ then the sequence $(\curState{\matchine}{t}{i})_{i}$ follows a Markov model $M = (\Minit, \Mtrans)$. 
Let us define
\begin{dmath*}
\subinc{n} \hiderel{=} \left\{v\hiderel{\in}\stateset^{*} \condi\sum_{i=0}^{\size{v}-2}\myshift(v_{i}, v_{i+1})+\kappa\hiderel{<} n \leq \sum_{i=0}^{\size{v}-1}\myshift(v_{i}, v_{i+1})+\kappa\right\}.
\end{dmath*}

From the fact that $\matchine$ is valid, we get $\lim_{\size{t}\to\infty}\tac{\matchine}(t) = \infty$ (Lemma \ref{lemmaShift}) and 
\begin{dmath*}
\lim_{n\to\infty}\sum_{t\in\alp^{n}} \frac{\size{t}}{\tac{\matchine}(t)}\piid(t) = \lim_{n\to\infty}\sum_{v\in\subinc{n}} \frac{\sum_{i=0}^{\size{v}-2}\myshift(v_{i},v_{i+1})}{\size{v}}\prob_{M}(v).
\end{dmath*}
Basically, we have that
\begin{dmath*}
\frac{\sum_{i=0}^{\size{v}-2}\myshift(v_{i},v_{i+1})}{\size{v}} = \sum_{d\in\stateset^{{2}}}\myshift(d_{0},d_{1})\frac{\size{v}_{d}}{\size{v}}
\end{dmath*}

The sequence $(v_{i}v_{i+1})_{i}$ follows a Markov model with states in $\stateset^{2}$. The same argument as in the proof of Theorem \ref{theoASExist} shows that the assumption of Lemma \ref{lemmaConv} is granted with $(v_{i}v_{i+1})_{i}$ and $\myshift$, which gives us that, for all $d\in\stateset^{{2}}$, 
\begin{dmath*}
\lim_{n\to\infty}\sum_{v\in\subinc{n}} \frac{\size{v}_{d}}{\size{v}}\prob_{M}(v)
= \lim_{k\to\infty}\sum_{v\in\stateset^{k}} \frac{\size{v}_{d}}{k}\prob_{M}(v)\\  
 = \lim_{k\to\infty}\sum_{v\in\stateset^{k}} \frac{\size{v}_{d_{0}}}{k}\times\frac{\size{v}_{d}}{\size{v}_{d_{0}}}\prob_{M}(v)\\
= \alpha_{d_{0}}\Mtrans(d_{0}, d_{1})
\end{dmath*}

Finally we have
\begin{dgroup*}
\begin{dmath*}
\lim_{n\to\infty}{\sum_{t\in\alp^{n}} \frac{\size{t}}{\tac{\matchine}(t)}\piid(t)}
= \lim_{n\to\infty}\sum_{v\in\subinc{n}}\frac{\sum_{i=0}^{\size{v}-2}\myshift(v_{i},v_{i+1})}{\size{v}}\prob_{M}(v)
\end{dmath*}
\begin{dmath*}
= \sum_{d\in\stateset^{{2}}}\myshift(d_{0},d_{1})\lim_{n\to\infty}\sum_{v\in\subinc{n}} \myshift(d_{0},d_{1})\frac{\size{v}_{d}}{\size{v}}\prob_{M}(v)
\end{dmath*}
\begin{dmath*}
= \sum_{d\in\stateset^{{2}}}\myshift(d_{0},d_{1})\alpha_{d_{0}}\Mtrans(d_{0}, d_{1})
\end{dmath*}
\begin{dmath*}
= \sum_{d_{0}\in\stateset} \alpha_{d_{0}}\sum_{d_{1}\in\stateset}\myshift(d_{0},d_{1})\Mtrans(d_{0}, d_{1})
\end{dmath*}
\end{dgroup*}

With Theorem \ref{theoiid}, we have that 
\begin{dmath*}
\sum_{d_{1}\in\stateset}\myshift(d_{0},d_{1})\Mtrans(d_{0}, d_{1})
=\sum_{d_{1}\in\stateset}\myshift(d_{0},d_{1})\frac{\displaystyle\smashoperator[r]{\sum_{{x, \trans(d_{0},x) = d_{1}}}} \piid(x)}{\displaystyle\smashoperator[r]{\sum_{{x,\trans(d_{0},x) \neq \sink}}} \piid(x)}
=\expstate{d_{0}}
\end{dmath*}

\end{proof}

\section{Withdrawing inefficient states}\label{s:states}

We shall see that some states of a matching machine may be removed without decreasing its asymptotic speed under a given iid model.

\subsection{Redirecting transitions}

\begin{theorem}\label{theoRedir}
	Let $\piid$ be an iid model, $t$ a text drawn from $\piid$, $\matchine=(\stateset, \init, \prematchset, \next, \trans, \shift)$ be a $w$-matching machine and $\sta$ and $\stb$ be two states which are such that, under the notations of Section \ref{secFull},
	\begin{itemize}
	\item for all states $r$ and all symbols $x$ and $y$, $\trans(r,x) = \trans(r,y) \Rightarrow \shift(r,x) = \shift(r,y)$;
	\item the sequence of states of the generic algorithm on the input $(\matchine, t)$ follows a Markov model $M=(\Minit, \Mtrans)$;
	\item the sequence of states of the generic algorithm on the input $(\redir{\matchine}{\sta}{\stb}, t)$ follows a Markov model $\moda=(\Minit[\moda], \Mtrans[\moda])$ which is such that
	\begin{itemize}
	\item if $\stb\neq\init$ then $\Minit[\moda]=\Minit$, otherwise $\Minit[\moda](\sta) = 1$ and $\Minit[\moda](s) = 0$ for all states $s\neq\sta$,
	\item $\Mtrans[\moda](r,s) = \Mtrans(r,s)$ for all states $s\neq\sta$,
	\item $\Mtrans[\moda](r,\sta) = \Mtrans(r,\sta) + \Mtrans(r,\stb)$;
	\end{itemize}
	\item the sequence of states of the generic algorithm on the input $(\redir{\matchine}{\stb}{\sta}, t)$ follows a Markov model $\modb=(\Minit[\modb], \Mtrans[\modb])$ which is such that
	\begin{itemize}
	\item if $\sta\neq\init$ then $\Minit[\modb]=\Minit$, otherwise $\Minit[\modb](\stb) = 1$ and $\Minit[\modb](s) = 0$ for all states $s\neq\stb$,
	\item $\Mtrans[\modb](r,s) = \Mtrans(r,s)$ for all states $s\neq\stb$,
	\item $\Mtrans[\modb](r,\stb) = \Mtrans(r,\sta) + \Mtrans(r,\stb)$.
	\end{itemize}
	\end{itemize}
	We have
	\begin{dmath*}
	\as{\piid}{\matchine}\leq\max \{\as{\piid}{\redir{\matchine}{\sta}{\stb}}, \as{\piid}{\redir{\matchine}{\stb}{\sta}}\}.
	\end{dmath*}
\end{theorem}

\begin{proof}
Under the assumptions of the theorem, the sequence $(\curState{\matchine}{t}{i})_{i}$ follows a Markov model $M = (\Minit, \Mtrans)$ and any transition from a state $r$ to a state $s$ (whatever the symbol read from the text) is associated to a unique shift which will be referred to as $\myshift(r,s)$. 

By defining the  set $\subinc{n}$ as  
\begin{dmath*}
\subinc{n} = \{v\hiderel{\in}\stateset^{*} \hiderel{\condi} \sum_{i=0}^{\size{v}-3}\myshift(v_{i},v_{i+1})+\kappa\hiderel{<} n \hiderel{\leq} \sum_{i=0}^{\size{v}-2}\myshift(v_{i},v_{i+1})+\kappa\},
\end{dmath*}
we have
\begin{dmath*}\label{eqDefAS0}
\as{\piid}{\matchine} = \lim_{n\to\infty}\sum_{t\in\alp^{n}} \frac{\size{t}}{\tac{\matchine}(t)}\prob_{\piid}(t)
= \lim_{n\to\infty}\sum_{v\in\subinc[\size{w}]{n}} \frac{\sum_{i=0}^{\size{v}-2}\myshift(v_{i},v_{i+1})}{\size{v}}\proba{M}(v).
\end{dmath*}

A similar argument as that of the proof of Lemma \ref{lemmaConv} shows that
\begin{dmath}\label{eqDefAS1}
\as{\piid}{\matchine} =  \lim_{k\to\infty}\sum_{v\in\stateset^{k}} \frac{\sum_{i=0}^{\size{v}-2}\myshift(v_{i},v_{i+1})}{\size{v}}\proba{M}(v).
\end{dmath}

Let now consider the Markov chain $V=(V_{i})_{i}$ where $V_{0}=\init$ and, for all $i\geq 0$, $\p\{V_{i+1} = s  \condi V_{i} = r\} = \Mtrans(r,s)$. The chain $V$ models the execution process of the generic algorithm on the input $(\matchine, t)$ with $t$ iid. Let us rewrite Equation \ref{eqDefAS1} as
\begin{dmath*}\label{eqDefAS2}
\as{\piid}{\matchine} =  \lim_{k\to\infty} \expect\left(\frac{\sum_{i=0}^{k-1} \myshift(V_{i}, V_{i+1})}{k}\right).
\end{dmath*}

The set of states of $V$ (or of any Markov chain) may be partitioned in a unique way into the class $\mathcal{T}$ of its transient states, which may be empty, and a positive number $c$ of non-empty recurrent classes (i.e. closed communicating classes) $\set{C}_{1}$, $\set{C}_{2}$, \ldots, $\set{C}_{c}$.

For all $1\leq m\leq c$, we define the Markov chains $V^{(m)}=(V^{(m)}_{i})_{i}$ where $V^{(m)}_{0}\in\set{C}_{m}$ and, for all $i\geq 0$, $\p\{V_{i+1} = s  \condi V_{i} = r\} = \Mtrans(r,s)$. All the chains $V^{(m)}$ are irreducible. The asymptotic speed of the class $\set{C}_{m}$ is noted $\asC{m}{\piid}{\matchine}$ and defined as
\begin{dmath*}\label{eqDefAS3}
\asC{m}{\piid}{\matchine} =  \lim_{k\to\infty} \expect\left(\frac{\sum_{i=0}^{k-1} \myshift(V^{(m)}_{i}, V^{(m)}_{i+1})}{k}\right).
\end{dmath*}

For two subsets $\set{E}$ and $\set{F}$ of $\stateset$ and $r$ a state of $\stateset_{\setminus \set{E}}$, let $\fhn{\set{E}}{r}{\set{F}}$ be the probability for the random variable $V_{n}$ to be in $\set{F}$ without visiting any state of $\set{E}\cup\set{F}$ from $1$ to $n-1$, being given that $V_{0} = r$, namely
\begin{dmath*}
	\fhn{\set{E}}{r}{\set{F}} \hiderel{=} \p\{V_{n} \in \set{F}\mbox{ and } V_{k}\not\in\set{E}\cup\set{F}\mbox{ for all } 0\hiderel{<}k\hiderel{<}n  \condi V_{0} \hiderel{=} r\}.
\end{dmath*}
We also define
\begin{dmath*}
	\fh{\set{E}}{r}{\set{F}} =  \sum_{n=1}^{\infty}\fhn{\set{E}}{r}{\set{F}}.
\end{dmath*}

Starting from $\init$ (or any state), the chain $V$ may visit some transient states but goes to one or another recurrent class in a time which is a.s. finite. Then, it stays in this recurrent class indefinitely.
By writing $\fh{}{r}{\set{F}}$ for $\fh{\emptyset}{r}{\set{F}}$, we have
\begin{dmath}\label{eqASMix}
	\as{\piid}{\matchine} = \sum_{m=1}^{c}\fh{}{\init}{\set{C}_{m}}\asC{m}{\piid}{\matchine}.
\end{dmath}

Since the chain $V$ ends up in a recurrent class with probability $1$, the law of total probability gives us that 
\begin{dmath*}\label{eqSumCoeff}
\sum_{m=1}^{c}\fh{}{\init}{\set{C}_{m}}=1.
\end{dmath*}

In order to prove the inequality of the theorem, we have to distinguish different cases according to which classes the states $\sta$ and $\stb$ belong.

\subsubsection*{Case 1 --  $\sta$ and $\stb$ are both transient}
Since $\sta$ and $\stb$ are reachable (from our implicit assumption), the state $\init$, which leads to $\sta$ and $\stb$, is transient as well. 

For all subsets $\set{S}\subset\stateset$ and all states $r\in\set{S}$, let $\fon{\set{S}}{r}$ denote the probability that $V_{n}=r$ is the last state of $\set{S}$ occurring in $V$, conditioned on  $V_{0}\in\set{S}$, namely
\begin{dmath*}
	\fon{\set{S}}{r} \hiderel{=} \p\{V_{n} \hiderel{=} r \mbox{ and } V_{k}\not\in\set{S}\mbox{ for all }k\hiderel{>}n  \condi V_{0}\in\set{S}\}.
\end{dmath*}
We also define
\begin{dmath*}
	\fo{\set{S}}{r} =  \sum_{n=1}^{\infty}\fon{\set{S}}{r}.
\end{dmath*}

Let us remark that if $\set{S}$ contains any recurrent state reachable from $r$ then both $\fon{\set{S}}{r}$ and $\fo{\set{S}}{r}$ are zero.

We have 
\begin{dmath}\label{eqGenC1}
\fh{}{\init}{\set{C}_{m}} = \fh{\{\sta,\stb\}}{\init}{\set{C}_{m}}+\fh{}{\init}{\{\sta,\stb\}}\left[\fo{\{\sta,\stb\}}{\sta}\fh{\{\sta,\stb\}}{\sta}{\set{C}_{m}}+\fo{\{\sta,\stb\}}{\stb}\fh{\{\sta,\stb\}}{\stb}{\set{C}_{m}}\right].
\end{dmath}

Substituting the coefficients of Equation \ref{eqGenC1} in Equation \ref{eqASMix}, gives us
\begin{dmath*}\label{eqGenC2}
\as{\piid}{\matchine} = \sum_{m=1}^{c}\left(\fh{\{\sta,\stb\}}{\init}{\set{C}_{m}}+\fh{}{\init}{\{\sta,\stb\}}\left[\fo{\{\sta,\stb\}}{\sta}\fh{\{\sta,\stb\}}{\sta}{\set{C}_{m}}+\fo{\{\sta,\stb\}}{\stb}\fh{\{\sta,\stb\}}{\stb}{\set{C}_{m}}\right]\right)\asC{m}{\piid}{\matchine}
= \sum_{m=1}^{c}\fh{\{\sta,\stb\}}{\init}{\set{C}_{m}}\asC{m}{\piid}{\matchine}+\fh{}{\init}{\{\sta,\stb\}}\left[\fo{\{\sta,\stb\}}{\sta}\sum_{m=1}^{c}\fh{\{\sta,\stb\}}{\sta}{\set{C}_{m}}\asC{m}{\piid}{\matchine}+\fo{\{\sta,\stb\}}{\stb}\sum_{m=1}^{c}\fh{\{\sta,\stb\}}{\stb}{\set{C}_{m}}\asC{m}{\piid}{\matchine}\right].
\end{dmath*}

Since both $\sta$ and $\stb$ are transient, there exists almost surely an integer $n$ such that neither $\sta$ nor $\stb$ occurs after $n$. Altogether with the law of total probability, this implies that 
\begin{dmath}\label{eqConvG}
\fo{\{\sta,\stb\}}{\sta}+\fo{\{\sta,\stb\}}{\stb} = 1.
\end{dmath}

Let now consider the matching machines $\redir{\matchine}{\sta}{\stb}$ and $\redir{\matchine}{\stb}{\sta}$. Under the assumptions of the theorem, the states of the generic algorithm follows the Markov models $\moda$ and $\modb$ on the inputs $(\redir{\matchine}{\sta}{\stb}, t)$ and $(\redir{\matchine}{\stb}{\sta}, t)$. Still from these assumptions, the models $\moda$
and $\modb$ differ with  $M$ only in the probability transitions ending on $\sta$ and on $\stb$ respectively. By putting $\typa{\fhsgn}$ and $\typb{\fhsgn}$ for the analogs of $\fhsgn$ with $\moda$
and $\modb$ respectively, we have, for all $1\leq m \leq c$,
\begin{itemize}
	\item $\asC{m}{\piid}{\redir{\matchine}{\sta}{\stb}} = \asC{m}{\piid}{\redir{\matchine}{\stb}{\sta}}=\asC{m}{\piid}{\matchine}$,
	\item $\fh{\{\sta,\stb\}}{\init}{\set{C}_{m}} = \fha{\sta}{\init}{\set{C}_{m}} =\fhb{\stb}{\init}{\set{C}_{m}}$,
	\item $\fh{}{\init}{\{\sta,\stb\}} = \fha{}{\init}{\sta} =\fhb{}{\init}{\stb}$,
	\item $\fh{\{\sta,\stb\}}{\sta}{\set{C}_{m}} = \fha{\sta}{\sta}{\set{C}_{m}}$,
	\item $\fh{\{\sta,\stb\}}{\sta}{\set{C}_{m}} = \fhb{\stb}{\stb}{\set{C}_{m}}$.
\end{itemize}

It follows that
\begin{dmath*}
\fha{}{\init}{\set{C}_{m}} = \fha{\sta}{\init}{\set{C}_{m}}+\fha{}{\init}{\sta}\fha{\sta}{\sta}{\set{C}_{m}}
=\fh{\{\sta,\stb\}}{\init}{\set{C}_{m}}+\fh{}{\init}{\{\sta,\stb\}}\fh{\{\sta,\stb\}}{\sta}{\set{C}_{m}}
\end{dmath*}
and
\begin{dmath*}
\fhb{}{\init}{\set{C}_{m}} = \fhb{\stb}{\init}{\set{C}_{m}}+\fhb{}{\init}{\stb}\fhb{\stb}{\stb}{\set{C}_{m}}
=\fh{\{\sta,\stb\}}{\init}{\set{C}_{m}}+\fh{}{\init}{\{\sta,\stb\}}\fh{\{\sta,\stb\}}{\stb}{\set{C}_{m}}.
\end{dmath*}

Considering Equation \ref{eqASMix} for the asymptotic speeds of $\redir{\matchine}{\sta}{\stb}$ and $\redir{\matchine}{\stb}{\sta}$ and the relations just above, leads to
\begin{dmath*}
\as{\piid}{\redir{\matchine}{\sta}{\stb}} = \sum_{m=1}^{c}\fha{}{\init}{\set{C}_{m}} \asC{m}{\piid}{\redir{\matchine}{\sta}{\stb}}
= \sum_{m=1}^{c}\fh{\{\sta,\stb\}}{\init}{\set{C}_{m}}\asC{m}{\piid}{\matchine}+\fh{}{\init}{\{\sta,\stb\}}\sum_{m=1}^{c}\fh{\{\sta,\stb\}}{\sta}{\set{C}_{m}}\asC{m}{\piid}{\matchine}
\end{dmath*}
and
\begin{dmath*}
\as{\piid}{\redir{\matchine}{\stb}{\sta}} = \sum_{m=1}^{c}\fhb{}{\init}{\set{C}_{m}} \asC{m}{\piid}{\redir{\matchine}{\stb}{\sta}}
= \sum_{m=1}^{c}\fh{\{\sta,\stb\}}{\init}{\set{C}_{m}}\asC{m}{\piid}{\matchine}+\fh{}{\init}{\{\sta,\stb\}}\sum_{m=1}^{c}\fh{\{\sta,\stb\}}{\stb}{\set{C}_{m}}\asC{m}{\piid}{\matchine}.
\end{dmath*}

Since a convex combination is smaller than the greatest of its elements, Equation \ref{eqConvG} gives us

\begin{dmath*}
\fo{\{\sta,\stb\}}{\sta}\sum_{m=1}^{c}\fh{\{\sta,\stb\}}{\sta}{\set{C}_{m}}\asC{m}{\piid}{\matchine}+\fo{\{\sta,\stb\}}{\stb}\sum_{m=1}^{c}\fh{\{\sta,\stb\}}{\stb}{\set{C}_{m}}\asC{m}{\piid}{\matchine}
\leq\max\left\{\sum_{m=1}^{c}\fh{\{\sta,\stb\}}{\sta}{\set{C}_{m}}\asC{m}{\piid}{\matchine}, \sum_{m=1}^{c}\fh{\{\sta,\stb\}}{\stb}{\set{C}_{m}}\asC{m}{\piid}{\matchine}\right\}
\end{dmath*}
and, finally,
\begin{dmath*}
\as{\piid}{\matchine}
\leq\max\left\{\as{\piid}{\redir{\matchine}{\sta}{\stb}}, \as{\piid}{\redir{\matchine}{\stb}{\sta}}\right\}.
\end{dmath*}

\subsubsection*{Case 2 -- $\sta$ is transient and $\stb$ is recurrent}
Let $\set{C}_{k}$ be the recurrent class to which $\stb$ belongs. We distinguish between two sub-cases according to whether or not $\set{C}_{k}$ is the only recurrent class reachable from $\stb$.
\subcase{Case 2a -- $\fh{}{\sta}{\set{C}_{k}}<1$}
It implies that $\sta$ leads to a recurrent class $\set{C}_{\ell}$ with $\ell\neq k$. Let us consider the Markov chain $\typa{V}=(\typa{V}_{i})_{i}$ where $\typa{V}_{0}=\init$ and, for all $i\geq 0$, $\p\{\typa{V}_{i+1} = s  \condi \typa{V}_{i} = r\} = \Mtrans[\moda](r,s)$. The set of states of  $\typa{V}$ is $\setMinus{\stateset}{\{\stb\}}$. Redirecting all the transitions that end with $\stb$, to $\sta$ makes all the states of $\setMinus{\set{C}_{k}}{\{\stb\}}$ transient in $\typa{V}$. In particular, the part of the asymptotic speed which comes from $\set{C}_{k}$ in $\as{\piid}{\matchine}$ just vanishes in $\as{\piid}{\redir{\matchine}{\sta}{\stb}}$.
For all $m\neq k$, the different ways of reaching $\set{C}_{m}$ from $\sta$ in the chain $\typa{V}$ may be split into the ways which follows a redirected transition and the ways which don't.
Under the theorem's assumptions, any path from $\sta$ to $\set{C}_{m}$ in $\typa{V}$ which contains no redirected transition has the same probability as in the chain $V$. Reciprocally, a path from $\sta$ to $\set{C}_{m}$ in $V$ contains no state of $\set{C}_{k}$, thus no transition which is redirected in $\typa{V}$. In other words, the probability of reaching $\set{C}_{m}$ in $\typa{V}$ without following a redirected transition being given that we start at $\sta$, is exactly $\fh{}{\sta}{\set{C}_{m}}$.

On the other hand, since $\stb\in\set{C}_{k}$ and $\set{C}_{k}$ is a recurrent class of $V$, the probability of following at least a redirected transition in a path being given that the path starts from $\sta$, is $\fh{}{\sta}{\set{C}_{k}}$. Moreover, since all the redirected transitions end at $\sta$, we have that

\begin{dmath*}
\fha{}{\sta}{\set{C}_{m}} = \fh{}{\sta}{\set{C}_{m}} + \fh{}{\sta}{\set{C}_{k}} \fha{}{\sta}{\set{C}_{m}} \condition{ thus }
\end{dmath*}
 \begin{dmath*}
\fha{}{\sta}{\set{C}_{m}} = \frac{\fh{}{\sta}{\set{C}_{m}}}{1-\fh{}{\sta}{\set{C}_{k}}} \hiderel{=} \frac{\fh{}{\sta}{\set{C}_{m}}}{\sum_{\substack{\ell=1\\ \ell\neq k}}^{c}\fh{}{\sta}{\set{C}_{m}}}\condition{for all $m\neq k$.}
\end{dmath*}

For all $1\leq m \leq c$ with $m\neq k$ and since  the paths of $\set{C}_{m}$ and those from $\init$ to $\sta$ or to a state of $\set{C}_{m}$, never visit $\stb$, we have 
\begin{itemize}
	\item $\asC{m}{\piid}{\redir{\matchine}{\sta}{\stb}} = \asC{m}{\piid}{\matchine}$,
	\item $\fha{\sta}{\init}{\set{C}_{m}} = \fh{\sta}{\init}{\set{C}_{m}}$,
	\item $\fha{}{\init}{\sta} = \fh{}{\init}{\sta}$.
\end{itemize}

Conversely, redirecting all the transitions that end with $\sta$, to $\stb$ increases the part of the asymptotic speed which comes from $\set{C}_{k}$.
We have, for all $1\leq m \leq c$,
\begin{itemize}
	\item $\asC{m}{\piid}{\redir{\matchine}{\stb}{\sta}}=\asC{m}{\piid}{\matchine}$,
	\item $\fhb{}{\init}{\set{C}_{m}} = \fh{}{\init}{\set{C}_{m}}-\fh{\sta}{\init}{\set{C}_{m}}$ if $m\neq k$,
	\item $\fhb{}{\init}{\set{C}_{k}}=\fh{}{\init}{\set{C}_{k}}+\sum_{\substack{\ell=1\\ \ell\neq k}}^{c}\fh{}{\sta}{\set{C}_{m}}$.
\end{itemize}

For all $1\leq m \leq c$, we have 
\begin{dmath*}
\fh{}{\init}{\set{C}_{m}} = \fh{\sta}{\init}{\set{C}_{m}}+\fh{}{\init}{\sta}\fh{}{\sta}{\set{C}_{m}}
\end{dmath*}
and, for $1\leq m \leq c$ with $m\neq k$,
\begin{dmath*}
\fha{}{\init}{\set{C}_{m}} = \fha{\sta}{\init}{\set{C}_{m}}+\fha{}{\init}{\sta}\fha{}{\sta}{\set{C}_{m}}.
\end{dmath*}

With Equation \ref{eqASMix}, it implies that
\begin{dgroup*}
\begin{dmath}\label{eqCas2b1}
\as{\piid}{\matchine} = \sum_{m=1}^{c}\fh{\sta}{\init}{\set{C}_{m}} \asC{m}{\piid}{\matchine}+\fh{}{\init}{\sta}\sum_{m=1}^{c}\fh{}{\sta}{\set{C}_{m}}\asC{m}{\piid}{\matchine},
\end{dmath}
\begin{dmath}\label{eqCas2b2}
\as{\piid}{\redir{\matchine}{\sta}{\stb}} = \sum_{m=1}^{c}\fh{\sta}{\init}{\set{C}_{m}} \asC{m}{\piid}{\matchine}+\fh{}{\init}{\sta}\frac{\sum_{\substack{\ell=1\\ \ell\neq k}}^{c}\fh{}{\sta}{\set{C}_{m}}\asC{m}{\piid}{\matchine}}{\sum_{\substack{\ell=1\\ \ell\neq k}}^{c}\fh{}{\sta}{\set{C}_{m}}}, 
\end{dmath}
\begin{dmath}\label{eqCas2b3}
\as{\piid}{\redir{\matchine}{\stb}{\sta}} = \sum_{m=1}^{c}\fh{\sta}{\init}{\set{C}_{m}} \asC{m}{\piid}{\matchine}+\fh{}{\init}{\sta}\asC{k}{\piid}{\matchine}.
\end{dmath}
\end{dgroup*}

We recall that $\sum_{m=1}^{c}\fh{}{\sta}{\set{C}_{m}} = 1$. From Equations \ref{eqCas2b1}, \ref{eqCas2b2} and \ref{eqCas2b3}, we get that
\begin{itemize}
	\item $\as{\piid}{\matchine}\geq \as{\piid}{\redir{\matchine}{\sta}{\stb}}$ if and only if $\asC{k}{\piid}{\matchine}\leq\frac{\sum_{\substack{\ell=1\\ \ell\neq k}}^{c}\fh{}{\sta}{\set{C}_{m}}\asC{m}{\piid}{\matchine}}{\sum_{\substack{\ell=1\\ \ell\neq k}}^{c}\fh{}{\sta}{\set{C}_{m}}}$,
	\item $\as{\piid}{\matchine}\leq \as{\piid}{\redir{\matchine}{\stb}{\sta}}$ if and only if $\asC{k}{\piid}{\matchine}\geq\frac{\sum_{\substack{\ell=1\\ \ell\neq k}}^{c}\fh{}{\sta}{\set{C}_{m}}\asC{m}{\piid}{\matchine}}{\sum_{\substack{\ell=1\\ \ell\neq k}}^{c}\fh{}{\sta}{\set{C}_{m}}}$.
\end{itemize}
In all cases, we have 
\begin{dmath*}
\as{\piid}{\matchine}\leq\max \{\as{\piid}{\redir{\matchine}{\sta}{\stb}}, \as{\piid}{\redir{\matchine}{\stb}{\sta}}\}.
\end{dmath*}

\subcase{Case 2b -- $\fh{}{\sta}{\set{C}_{k}}=1$}
It means that $\sta$ leads only to the recurrent class $\set{C}_{k}$. Redirecting all the transitions that end with $\stb$, to $\sta$ just replaces the recurrent class $\set{C}_{k}$ of $\mm$, by the recurrent class $\typa{\set{C}}_{k}$ of $\moda$, in which $\stb$ plays the role of $\sta$. Let $\asC{k}{\piid}{\redir{\matchine}{\sta}{\stb}}$ be the asymptotic speed of $\typa{\set{C}}_{k}$. We remark, first, that $\fha{}{\init}{\typa{\set{C}}_{k}} = \fh{}{\init}{\set{C}_{k}}$ and, second, that, for all $m\neq k$, $\asC{m}{\piid}{\redir{\matchine}{\sta}{\stb}} = \asC{m}{\piid}{\matchine}$.

Conversely, redirecting all the transitions that end with $\sta$, to $\stb$ does not change any recurrent class between $M$ and $\modb$. Moreover we have $\fha{}{\init}{\set{C}_{k}} = \fh{}{\init}{\set{C}_{k}}$ and, more generally, $\fha{}{\init}{\set{C}_{m}} = \fh{}{\init}{\set{C}_{m}}$ for all $1\leq m \leq c$. With Equation \ref{eqASMix}, we get that $\as{\piid}{\matchine}=\as{\piid}{\redir{\matchine}{\stb}{\sta}}$.

In short, we have $\as{\piid}{\redir{\matchine}{\sta}{\stb}}\geq\as{\piid}{\matchine}=\as{\piid}{\redir{\matchine}{\sta}{\stb}}$ if and only if $\asC{k}{\piid}{\redir{\matchine}{\sta}{\stb}}\geq\asC{k}{\piid}{\matchine}$, which leads to 
\begin{dmath*}
\as{\piid}{\matchine}
\leq\max\left\{\as{\piid}{\redir{\matchine}{\sta}{\stb}}, \as{\piid}{\redir{\matchine}{\stb}{\sta}}\right\}.
\end{dmath*}
\subsubsection*{Case 3 -- $\sta$ is recurrent and $\stb$ is transient}
This case is perfectly symmetrical with Case 2.

\subsubsection*{Case 4 -- $\sta$ and $\stb$ are both recurrent}
We have to distinguish between two sub-cases according to whether $\sta$ and $\stb$ are in the same recurrent class.
\subcase{Case 4a -- $\sta$ and $\stb$ are in two different recurrent classes} Let $\set{C}_{k}$ be the class of $\sta$ and $\set{C}_{\ell}$ be the class of $\stb$. Redirecting all the transitions that end with $\stb$, to $\sta$ makes all the states of $\set{C}_{\ell}$ transient (i.e. $\set{C}_{\ell}$ is not a recurrent class of $\moda$). Moreover, since all the states of $\set{C}_{\ell}$ lead to $\sta\in\set{C}_{k}$ in $\moda$, we have $\fha{}{\init}{\set{C}_{k}} = \fh{}{\init}{\set{C}_{k}}+\fh{}{\init}{\set{C}_{\ell}}$ and $\fha{}{\init}{\set{C}_{m}} = \fh{}{\init}{\set{C}_{m}}$ for all $m$ different from both $k$ and $\ell$. Redirecting all the transitions that end with $\sta$, to $\stb$ leads to symmetrical considerations. It follows that we have
\begin{itemize}
	\item $\as{\piid}{\redir{\matchine}{\sta}{\stb}}\geq\as{\piid}{\matchine}$ if and only if $\asC{k}{\piid}{\matchine}\geq\asC{\ell}{\piid}{\matchine}$,
	\item $\as{\piid}{\redir{\matchine}{\stb}{\sta}}\geq\as{\piid}{\matchine}$ if and only if $\asC{k}{\piid}{\matchine}\leq\asC{\ell}{\piid}{\matchine}$. 
\end{itemize}
We get again
\begin{dmath*}
\as{\piid}{\matchine}\leq\max \{\as{\piid}{\redir{\matchine}{\sta}{\stb}}, \as{\piid}{\redir{\matchine}{\stb}{\sta}}\}.
\end{dmath*}

\subcase{Case 4b -- $\sta$ and $\stb$ belong to a same recurrent class $\set{C}_{k}$}
 
Redirecting transitions toward $\sta$ or $\stb$ does not change neither the asymptotic speeds of the recurrent classes $(\set{C}_{m})_{m\neq k}$, nor the probabilities to end up in one of these classes from $\init$. We start by focusing on the class $\set{C}_{k}$. 

Since $V^{(k)}$ is irreducible, assuming that $V^{(k)}_{0} = \sta$ is convenient and does not influence $\asC{k}{\piid}{\matchine}$.

Let us define $\rvp{n}$ as the position of the $n^{\mbox{\tiny th}}$ occurrence of $\sta$ or $\stb$ in $V^{(k)}$. Namely $(\rvp{n})_{n}$ is such that $\rvp{0}=0$ (since we assume $V^{(k)}_{0} = \sta$) and for all $n\geq 0$,
\begin{itemize}
\item $V^{(k)}_{\rvp{n}}\in\{\sta, \stb\}$,
\item for all $\rvp{n}< i < \rvp{n+1}$, $V^{(k)}_{i}\not\in\{\sta, \stb\}$.
\end{itemize}
Let $\rvia{i}$ (resp. $\rvib{i}$) be such that $\rvp{\rvia{i}}$ (resp. $\rvp{\rvib{i}}$) is the position of the $i^{\mbox{\tiny th}}$ occurrence of $\sta$ (resp. of $\stb$) in  $V^{(k)}$. 
For all positions $i$, we put $\rvnia{i}$ (resp. $\rvnib{i}$) for the number of occurrences of $\sta$ (resp. of $\stb$) in $V^{(k)}_{0}, \ldots, V^{(k)}_{i}$. We set $\rvn{i}=\rvnia{i}+\rvnib{i}$ and we define the binary random variables $\rvea{i}$ and  $\rveb{i}$ as 
\begin{dmath*}
\rvea{i} \hiderel{=} \left\{\begin{array}{ll} 1 &\mbox{if } \rvia{\rvnia{i}}>\rvib{\rvnib{i}},\\ 0 & \mbox{otherwise,}\end{array}\right. \mbox{ and } \rveb{i} \hiderel{=} \left\{\begin{array}{ll} 1 &\mbox{if } \rvia{\rvnia{i}}<\rvib{\rvnib{i}},\\ 0 & \mbox{otherwise.}\end{array}\right.
\end{dmath*}

By setting $\rvr{i} = \sum_{j=\rvp{i}}^{\rvp{i+1}-1} \myshift(V^{(k)}_{j}, V^{(k)}_{j+1})$ we get, for all integers $n$,
\begin{dmath*}
\frac{\sum_{i=0}^{\rvn{n}-1}\rvr{i}}{n} \hiderel{\leq} \frac{\sum_{j=0}^{n}\myshift(V^{(k)}_{j}, V^{(k)}_{j+1})}{n} \hiderel{\leq} \frac{\sum_{i=0}^{\rvn{n}}\rvr{i}}{n}.
\end{dmath*}

Decomposing the sums above leads to
\begin{dmath*}
\frac{\sum_{i=0}^{\rvnia{n}-\rvea{n}}\rvr{\rvia{i}}}{n} +\frac{\sum_{i=0}^{\rvnib{n}-\rveb{n}}\rvr{\rvib{i}}}{n}\hiderel{\leq} \frac{\sum_{j=0}^{n}\myshift(V^{(k)}_{j}, V^{(k)}_{j+1})}{n} \hiderel{\leq} \frac{\sum_{i=0}^{\rvnia{n}}\rvr{\rvia{i}}}{n} +\frac{\sum_{i=0}^{\rvnib{n}}\rvr{\rvib{i}}}{n}.
\end{dmath*}

Let now consider $\redir{\matchine}{\sta}{\stb}$ and the corresponding Markov model $\moda$. We define the Markov chain $\rva=(\rva_{i})_{i}$ with $\rva_{0}=\init$ and, for all $i\geq 0$, $\p\{\rva_{i+1} = s  \condi \rva_{i} = r\} = \Mtrans[\moda](r,s)$. 
By construction, the chain $\rva$ contains all the recurrent classes $(\set{C}_{m})_{m\neq k}$. Moreover, since $\sta$ and $\stb$ are both in a recurrent class of $V$, all the states which were transient with $V$ are still transient in $\rva$. In particular, for all $m\neq k$, no path from $\init$ to a recurrent class $\set{C}_{m}$ contains a redirected transition. We have 
\begin{dmath*}
\fh{}{\init}{\set{C}_{m}} = \fh{\{\sta, \stb\}}{\init}{\set{C}_{m}} = \fha{\{\sta, \stb\}}{\init}{\set{C}_{m}} = \fha{}{\init}{\set{C}_{m}}.
\end{dmath*}

Since all the states that lead to $\sta$ in the chain $V$, still lead to $\sta$ in the chain  $\rva$, $\rva$ contains a recurrent class $\typa{\set{C}}_{k}$ to which $\sta$ belongs.
Let $q\neq\sta$ be a state of $\set{C}_{k}$. Several possibilities arise:
\begin{itemize}
	\item if $q$ is reachable from $\sta$ in $\rva$ then $q\in\typa{\set{C}}_{k}$;
	\item if $q$ is reachable from $\init$ but not from $\sta$ then $q$ is transient in $\rva$;
	\item if $q$ is not reachable from $\init$ then it is not a state of $\rva$.
\end{itemize}
In short, the chain $\rva$ contains all the recurrent classes $(\set{C}_{m})_{m\neq k}$, a non-empty recurrent class  $\typa{\set{C}}_{k}\subset\set{C}_{k}$, a set of transient states which contains that of $V$. We have
\begin{dmath*}
\fha{}{\init}{\typa{\set{C}}_{k}} = \fh{}{\init}{\set{C}_{k}}.
\end{dmath*}
By defining, for all $i\geq 0$, 
\begin{itemize}
\item $\rvpa{i}$ as the position of the $i^{\mbox{\tiny th}}$ occurrence of $\sta$ in $\rva^{(k)}$,  
\item $\rvnia{i}$ as the number of occurrences of $\sta$ in $\rva^{(k)}_{0}, \ldots, \rva^{(k)}_{i}$, 
\item $\rvra{i}$ as $\rvra{i} = \sum_{j=\rvpa{i}}^{\rvpa{i+1}-1} \myshift(\rva^{(k)}_{j}, \rva^{(k)}_{j+1})$, 
\end{itemize}
we have for all $n>0$,

\begin{dmath*}
\frac{\sum_{i=0}^{\rvna{n}-1}\rvra{i}}{n} \hiderel{\leq} \frac{\sum_{j=0}^{n}\myshift(\rva^{(k)}_{j}, \rva^{(k)}_{j+1})}{n} \hiderel{\leq} \frac{\sum_{i=0}^{\rvna{n}}\rvra{i}}{n}.
\end{dmath*}

By setting $\rvca{i} = \rvpa{i+1}-\rvpa{i}$ for all $i\geq 0$, we have
\begin{dmath*}
\frac{\sum_{i=0}^{\rvna{n}-1}\rvca{i}}{\rvna{n}}\hiderel{\leq} \frac{n}{\rvna{n}}\hiderel{\leq}\frac{\sum_{i=0}^{\rvna{n}}\rvca{i}}{\rvna{n}}.
\end{dmath*}

The argument is essentially the same as for the proof of the \emph{renewal reward theorem}. All the $\rvca{i}$ are independent and identically distributed (the Markov chain $\rva^{(k)}$ is homogeneous and $\rva^{(k)}_{\rvpa{i}}=\sta$ for all $i$). We put $\expect(\rvca{})$ for their expectation. Moreover, the chain $\rva^{(k)}$ is irreducible an contains a finite number of states, which are thus all positive recurrent. In particular, the mean recurrence time for $\sta$ is finite. Since, whatever $i$, the random variable $\rvca{i}$ accounts for the recurrence time of $\sta$, the expectation $\expect(\rvca{})$ is finite, which implies that $\lim_{n\to\infty}\rvna{n} = \infty$.
The strong law of large number gives us that 
\begin{dmath*}
\lim_{n\to\infty} \frac{n}{\rvna{n}} = \expect(\rvca{})\condition[]{a.s.}
\end{dmath*}

In the same way, the random variables $\rvra{i}$ are independent and identically distributed. Moreover, since $\expect(\rvca{})$ is finite and $\myshift$ is bounded, the identically distributed random variables $\rvra{i}$ have a finite expectation $\expect(\rvra{})$. Applying again the strong law of large number leads to
\begin{dmath*}
\lim_{n\to\infty}\frac{\sum_{i=0}^{\rvna{n}}\rvra{i}}{n} = \lim_{n\to\infty}\frac{\sum_{i=0}^{\rvna{n}}\rvra{i}}{\rvna{n}}\times\frac{\rvna{n}}{n}
=\frac{\expect(\rvra{})}{\expect(\rvca{})}\condition[]{a.s.}
\end{dmath*}

From the bounded convergence theorem, we get that
\begin{dmath*}
\asC{k}{\piid}{\redir{\matchine}{\sta}{\stb}} =\lim_{n\to\infty} \expect\left(\frac{\sum_{j=0}^{n}\myshift(\rva^{(k)}_{j}, \rva^{(k)}_{j+1})}{n}\right)
 = \lim_{n\to\infty} \expect\left(\frac{\sum_{i=0}^{\rvna{n}}\rvra{i}}{n}\right) \hiderel{=}  \expect\left(\lim_{n\to\infty}\frac{\sum_{i=0}^{\rvna{n}}\rvra{i}}{n}\right)
 = \frac{\expect(\rvra{})}{\expect(\rvca{})}.
\end{dmath*}

The random variables $\rvc{\rvia{i}}$ are independent and identically distributed (with the same argument as above). Moreover, by construction, they follow the same distribution as the random variables $\rvca{i}$. Since all the transitions that go to $\stb$ in $M$, go to $\sta$ in $\moda$, starting with $\sta$ and ending at the first $\sta$ or $\stb$ in the chain $V^{(k)}$ is the same as starting with $\sta$ and ending at the first $\sta$ in $\rva^{(k)}$. The random variables $\rvc{\rvia{i}}$ have expectation $\expect(\rvca{})$.
In the same way, the random variables $(\rvr{\rvia{i}})_{{i}}$ are independent, identically distributed and follow the same distribution as the variables $(\rvra{i})_{i}$, thus have expectation $\expect(\rvra{})$.
The strong law of large numbers gives us
\begin{dgroup*}
\begin{dmath*}
 \lim_{n\to\infty}\frac{\sum_{i=0}^{\rvnia{n}}\rvr{\rvia{i}}}{\rvnia{n}} = \expect(\rvra{})\condition{a.s.,}
\end{dmath*}
\begin{dmath*}
 \lim_{n\to\infty}\frac{\sum_{i=0}^{\rvnia{n}}\rvc{\rvia{i}}}{\rvnia{n}} = \expect(\rvca{})\condition{a.s.}
\end{dmath*}
\end{dgroup*}

Moreover since the chain $V^{(k)}$ is irreducible, we have
\begin{dmath*}
 \lim_{n\to\infty}\frac{\rvnia{n}}{n} = \stfreqa\condition{a.s.,}
\end{dmath*}
where $\stfreqa$ is the probability of $\sta$ in the stationary distribution of $V^{(k)}$.

\begin{dmath*}
 \lim_{n\to\infty}\frac{\sum_{i=0}^{\rvnia{n}}\rvr{\rvia{i}}}{n} = \lim_{n\to\infty}\frac{\sum_{i=0}^{\rvnia{n}}\rvr{\rvia{i}}}{\rvnia{n}}\times\frac{\rvnia{n}}{n}
 =\expect(\rvra{})\stfreqa\condition{a.s.}
\end{dmath*}

From the bounded convergence theorem, we get that
\begin{dmath*}
 \lim_{n\to\infty}\expect\left(\frac{\sum_{i=0}^{\rvnia{n}}\rvr{\rvia{i}}}{n}\right) = \expect(\rvra{})\stfreqa
=\frac{\expect(\rvra{})}{\expect(\rvca{})}\expect(\rvca{})\stfreqa{}
 =\asC{k}{\piid}{\redir{\matchine}{\sta}{\stb}}\expect(\rvca{})\stfreqa{}.
\end{dmath*}

Symmetrically, we have
\begin{dmath*}
 \lim_{n\to\infty}\expect\left(\frac{\sum_{i=0}^{\rvnib{n}}\rvr{\rvib{i}}}{n}\right) =\asC{k}{\piid}{\redir{\matchine}{\stb}{\sta}}\expect(\rvcb{})\stfreqb{}.
\end{dmath*}

In order to prove that $\expect(\rvca{})\stfreqa+\expect(\rvcb{})\stfreqb{} = 1$, let us define the random variables $\rvc{i} = \rvp{i+1}-\rvp{i}$. 
We have 
\begin{dmath*}
\frac{\sum_{i=0}^{\rvn{n}}\rvc{i}}{n}= \frac{\sum_{i=0}^{\rvnia{n}}\rvc{\rvia{i}}}{n} + \frac{\sum_{i=0}^{\rvnib{n}}\rvc{\rvib{i}}}{n} 
 = \frac{\sum_{i=0}^{\rvnia{n}}\rvc{\rvia{i}}}{\rvnia{n}}\times\frac{\rvnia{n}}{n}+\frac{\sum_{i=0}^{\rvnib{n}}\rvc{\rvib{i}}}{\rvnib{n}}\times\frac{\rvnib{n}}{n}.
\end{dmath*}

Since the expectation of $\rvc{i}$ is smaller than the expected return times of  the positive recurrent states $\sta$ and $\stb$, it is finite. Since moreover
\begin{dmath*}
\frac{\sum_{i=0}^{\rvn{n}-1}\rvc{i}}{n}\hiderel{\leq} n \hiderel{\leq} \frac{\sum_{i=0}^{\rvn{n}}\rvc{i}}{n},
\end{dmath*}
we have
\begin{dmath*}
  \lim_{n\to\infty}\expect\left(\frac{\sum_{i=0}^{\rvn{n}}\rvc{i}}{n}\right) = 1
 = \lim_{n\to\infty}\expect\left(\frac{\sum_{i=0}^{\rvnia{n}}\rvc{\rvia{i}}}{\rvnia{n}}\times\frac{\rvnia{n}}{n}\right)+\lim_{n\to\infty}\expect\left(\frac{\sum_{i=0}^{\rvnib{n}}\rvc{\rvib{i}}}{\rvnib{n}}\times\frac{\rvnib{n}}{n}\right)
 = \expect(\rvca{})\stfreqa+\expect(\rvcb{})\stfreqb{}.
\end{dmath*}

The asymptotic speed of the recurrent class $\set{C}_{k}$ may be written as
\begin{dmath*}
\asC{k}{\piid}{\matchine} = \lim_{n\to\infty}\expect\left(\frac{\sum_{j=0}^{n}\myshift(V^{(k)}_{j}, V^{(k)}_{j+1})}{n}\right)
=\lim_{n\to\infty}\expect\left(\frac{\sum_{i=0}^{\rvnia{n}}\rvr{\rvia{i}}}{n} +\frac{\sum_{i=0}^{\rvnib{n}}\rvr{\rvib{i}}}{n}\right)
=\asC{k}{\piid}{\redir{\matchine}{\sta}{\stb}}\expect(\rvca{})\stfreqa{}+\asC{k}{\piid}{\redir{\matchine}{\stb}{\sta}}\expect(\rvcb{})\stfreqb.
\end{dmath*}
As a convex combination of $\asC{k}{\piid}{\redir{\matchine}{\sta}{\stb}}$ and $\asC{k}{\piid}{\redir{\matchine}{\stb}{\sta}}$, $\asC{k}{\piid}{\matchine}$ is smaller than their maximum. This last case leads again to 
\begin{dmath*}
\as{\piid}{\matchine}\leq\max \{\as{\piid}{\redir{\matchine}{\sta}{\stb}}, \as{\piid}{\redir{\matchine}{\stb}{\sta}}\}
\end{dmath*}
and ends the proof.
\end{proof}

\begin{corollary}\label{coroDup}
	Let $\piid$ be an iid model, $\matchine$ be a standard $w$-matching machine. If the states $\sta$ and $\stb$ are such that $\mem{\sta} = \mem{\stb}$ then we have 
	$$\as{\piid}{\matchine}\leq\max \{\as{\piid}{\redir{\matchine}{\sta}{\stb}}, \as{\piid}{\redir{\matchine}{\stb}{\sta}}\}$$
\end{corollary}
\begin{proof}
With Lemma \ref{lemmaStandard} and Theorem \ref{theoiid}, the sequences of states of an execution of $\matchine$ follows a Markov model which satisfies the assumptions of Theorem \ref{theoRedir}. 
From Lemma \ref{lemmaRedir},  $\redir{\matchine}{\sta}{\stb}$ and $\redir{\matchine}{\stb}{\sta}$ are still standard. Again with Theorem \ref{theoiid}, the corresponding sequences of states follows two Markov models $\moda$ and $\modb$, respectively, which, by construction, satisfy the assumptions of Theorem \ref{theoRedir}.
\end{proof}

\subsection{Minimal shift to a match - relevant states}

Let $\matchine = (\stateset, \init, \prematchset, \next, \trans, \shift)$ be a $w$-matching machine. A state $q\in\stateset$ is \defi{relevant} if it leads to a match transition reporting its current position, namely, if there exist a text $t$ and two indexes $i<j$ such that $\curState{\matchine}{t}{i} = q$, $\curState{\matchine}{t}{j}\in\prematchset$, $t_{\curPos{\matchine}{t}{j}+\next(\curState{\matchine}{t}{j})} = w_{\curState{\matchine}{t}{j}}$ and $\curShift{\matchine}{t}{k} = 0$ for all $i\leq k<j$. 
Under the implicit assumptions on matching machines (end of Section \ref{secGen}), all the pre-match states are relevant.

For all states $q\in\stateset$, we recursively define $\mnshft(q)$ as:
\begin{dmath*}
\mnshft(q) = \left\{\begin{array}{ll}
0 & \mbox{if }q\in\prematchset,\\
\min_{x\in\alp} \{\mnshft(\trans(q,x))+\shift(q,x)\} & \mbox{otherwise.}
\end{array}\right.
\end{dmath*}

\begin{remark}
If $\matchine$ is valid, then a state $q$ is relevant if and only if $\mnshft(q) = 0$.
\end{remark}

Let $\expan{\matchine}$ be the full memory expansion of $\matchine$ (Section \ref{secFull}). For all states $q$ of $\matchine$, we define $\exst(q)$ as the set comprising all the elements of $\readset{\ordmatchine}$ associated with $q$ in $\expan{\matchine}$, namely, $\exst(q) = \{H \condi(q,H)\in\expan{\stateset}\}$. If $\matchine$ is standard then for all states $q$, $\exst(q)$ is a singleton. 
A state $q$ of $\matchine$ is said \defi{consistent}, if all pairs $(H,H')$of elements of $\exst(q)$ verify the following properties
\begin{enumerate}
\item $\frst(H)=\frst(H')$,
\item for all $(i, x)\in H$, $i\geq\mnshft(q)\Rightarrow (i, x)\in H'$,
\end{enumerate}
where $\frst(H)$ is the set of position entries of the elements of $H$ (see Section \ref{secFull}).
In particular, all the states of a standard $w$-matching machine are consistent.

Two states $q$ and $q'$ are \defi{interchangeable} if they are both consistent and such that all pairs $(H, H')$ with $H\in\exst(q)$ and $H'\in\exst(q')$ verify the two properties just above.

\begin{lemma}\label{remMnsft}
Let $\matchine$ be a non-redundant $w$-matching machine in which all the states are consistent, and  $\sta$ and $\stb$ be two interchangeable states of $\matchine$.
\begin{enumerate}
	\item All the states of $\redir{\matchine}{\sta}{\stb}$ (resp. of $\redir{\matchine}{\stb}{\sta}$) are consistent.
	\item If $\matchine$ is valid then both $\redir{\matchine}{\sta}{\stb}$ and $\redir{\matchine}{\stb}{\sta}$ are valid.
	\item If, moreover, for all states $r$ and all symbols $x$ and $y$, $\trans(r,x) = \trans(r,y) \Rightarrow \shift(r,x) = \shift(r,y)$ then, for all iid model $\piid$, we have \begin{dmath*}\as{\piid}{\matchine}\leq\max \{\as{\piid}{\redir{\matchine}{\sta}{\stb}}, \as{\piid}{\redir{\matchine}{\stb}{\sta}}\}.\end{dmath*}
\end{enumerate}
\end{lemma}
\begin{proof}
Property 1 comes straightforwardly with the definitions of consistency and interchangeability.

Property 2 may be proved in the same way as Lemma \ref{lemmaRedir}.

In order to prove Property 3, we remark that, since $\matchine$ is non-redundant and $\sta$ and $\stb$ are interchangeable, both  $\redir{\matchine}{\sta}{\stb}$ and $\redir{\matchine}{\stb}{\sta}$ are non-redundant. With Theorem \ref{theoiid},  we get that, if $t$ is iid, the sequence of states parsed during an execution of the algorithm on the input $(t,\matchine)$ (resp. $(t,\redir{\matchine}{\sta}{\stb})$, $(t,\redir{\matchine}{\stb}{\sta})$) follows a Markov model $M=(\Minit, \Mtrans)$ (resp. $\moda=(\Minit[\moda], \Mtrans[\moda])$, $\modb=(\Minit[\modb], \Mtrans[\modb])$). Moreover, since for all states $r$ and $s$, we have 
\begin{dmath*}\Mtrans(r,s) = \displaystyle\smashoperator[r]{\sum_{x,\trans(r,x) = s}} \piid(x),\end{dmath*}
\begin{dmath*}\Mtrans[\moda](r,s) = \displaystyle\smashoperator[r]{\sum_{x,\redir{\trans}{\sta}{\stb}(r,x) = s}} \piid(x),\end{dmath*}
\begin{dmath*}\Mtrans[\modb](r,s) = \displaystyle\smashoperator[r]{\sum_{x ,\redir{\trans}{\stb}{\sta}(r,x) = s}} \piid(x),\end{dmath*}
and with the definition of $\redir{\matchine}{\sta}{\stb}$ and $\redir{\matchine}{\stb}{\sta}$, both $\moda$ and $\modb$ satisfy the assumptions of Theorem  \ref{theoRedir}.
If, moreover, for all states $r$ and all symbols $x$ and $y$, $\trans(r,x) = \trans(r,y) \Rightarrow \shift(r,x) = \shift(r,y)$, all the assumptions of Theorem \ref{theoRedir} are granted, which leads to Property 3.
\end{proof}

\begin{lemma}\label{lemmaMnsft2}
Let $\matchine$ be a valid, non-redundant and compact $w$-matching machine containing only consistent states. For all iid models $\piid$, there exists a $w$-matching machine $\matchine'$ such that 
\begin{enumerate}
	\item for all states $q\in\stateset'$, if $\next'(q)<\mnshft(q)$ then for all symbols $x$ and $y$, both $\trans'(q,x) = \trans'(q,y)$ and $\shift'(q,x) = \shift'(q,y)$;
	\item $\as{\piid}{\matchine'}\geq\as{\piid}{\matchine}$.
\end{enumerate}
\end{lemma}
\begin{proof}
Let us first remark that under the assumptions that $\matchine$ is valid and all its states consistent, if there exist two symbols $x$ and $y$ such that $\trans(q,x) \neq \trans(q,y)$  then $\shift(q,x) \neq \shift(q,y)$ (it can be proved is the same way as Lemma \ref{lemmaStandard}).

Let us assume that there exists a state $q\in\stateset$ such that both $\next(q)<\mnshft(q)$ and $\trans(q,x) \neq \trans(q,y)$. Since all the states are consistent, the states $\sta=\trans(q,x)$ and $\stb=\trans(q,y)$ are interchangeable. Both $\redir{\matchine}{\sta}{\stb}$ and $\redir{\matchine}{\sta}{\stb}$ are such that $\redir{\trans}{\sta}{\stb}(q,x) \neq \redir{\trans}{\sta}{\stb}(q,y)$. Lemma \ref{remMnsft} ensures that both  $\redir{\matchine}{\sta}{\stb}$ and $\redir{\matchine}{\sta}{\stb}$ contain only consistent states, are valid and such that \begin{dmath*}\as{\piid}{\matchine}\leq\max \{\as{\piid}{\redir{\matchine}{\sta}{\stb}}, \as{\piid}{\redir{\matchine}{\stb}{\sta}}\}.\end{dmath*}

We put $\matchine'$ for the machine with the greatest asymptotic speed among $\redir{\matchine}{\sta}{\stb}$ and $\redir{\matchine}{\sta}{\stb}$. If $\matchine'$ is such that for all states $q\in\stateset'$, if $\next'(q)<\mnshft(q)$ then for all symbols $x$ and $y$, both $\trans'(q,x) = \trans'(q,y)$ and $\shift'(q,x) = \shift'(q,y)$, the lemma is proved. Otherwise, we replace $\matchine$ by $\matchine'$ which still satisfies the assumptions of the lemma and has a greater asymptotic speed before iterating the same process. Since at each iteration, there is a state $q$ and two symbols $x$ and $y$ such that  $\trans(q,x) \neq \trans(q,y)$ and $\trans'(q,x) = \trans'(q,y)$, we eventually end with a machine $\matchine'$ with the desired property.
\end{proof}

Let $\matchine = (\stateset, \init, \prematchset, \next, \trans, \shift)$ be a $w$-matching machine verifying $\next(q)\geq\mnshft(q)$ for all states $q\in\stateset$. The $w$-matching machine $\posit{\matchine} = (\posit{\stateset}, \posit{\init}, \posit{\prematchset}, \posit{\next}, \posit{\trans}, \posit{\shift})$ is defined as
\begin{itemize}
	\item $\posit{\stateset} = \stateset$,
	\item $\posit{\init} = \init$,
	\item $\posit{\prematchset} = \prematchset$,
	\item $\posit{\next}(q) = \next(q)-\mnshft(q)$,
	\item $\posit{\trans}(q,x) = \trans(q,x)$,
	\item $\posit{\shift}(q,x) = \shift(q,x)-\mnshft(q)+\mnshft(\trans(q,x))$.
\end{itemize}
for all states $q \in \stateset$.

If $\matchine$ is such that $\next(q)\geq\mnshft(q)$ for all $q\in\stateset$, the quantities $\posit{\next}(q)$ and $\posit{\shift}(q,x)$ are non-negative for all states $q$ and all symbols $x$ (i.e. $\posit{\matchine}$ is well a $w$-matching machine).

\begin{remark}\label{remPlus}
 For all texts $t$, the sequences of accessed  positions coincide during the executions of the generic algorithm on the inputs $(\matchine, t)$ and $(\posit{\matchine}, t)$. In particular, $\matchine$ is valid if and only if  $\posit{\matchine}$ is valid and the asymptotic speeds of $\matchine$ and $\posit{\matchine}$ are equal.
\end{remark}

\begin{theorem}\label{theoMain}
Let $\matchine$ be a valid $w$-matching machine. For all iid models $\piid$, there exists a $w$-matching machine $\matchine_{\piid}$ with $\as{\piid}{\matchine}\leq\as{\piid}{\matchine_{\piid}}$ and which is
\begin{itemize}
	\item standard, 
	\item compact,
	\item valid,
	\item in which all the states are relevant,
	\item such that there is no pair of states $(q,q')$ such that $q\neq q'$ and $\mem[\matchine_{\piid}]{q} = \mem[\matchine_{\piid}]{q'}$.
\end{itemize}
\end{theorem}
\begin{proof}
	
With Proposition \ref{propComp}, there exists a standard, compact and valid $w$-matching machine $\matchine_{a}$ such that $\as{\piid}{\matchine}\leq\as{\piid}{\matchine_{a}}$.
Since the machine $\matchine_{a}$ satisfies the assumptions of Lemma \ref{lemmaMnsft2}, there exists a $w$-matching machine $\matchine_{b}$ which is
\begin{itemize}
	\item valid, 
	\item in which all the states are consistent, 
	\item with a asymptotic speed greater than $\matchine_{a}$,
	\item such that  for all states $q\in\stateset_{b}$, if $\next_{b}(q)<\mnshft(q)$ then for all symbols $x$ and $y$, both $\trans_{b}(q,x) = \trans_{b}(q,y)$ and $\shift_{b}(q,x) = \shift_{b}(q,y)$.
\end{itemize} 

The machine $\matchine_{b}$ is still non-redundant. It is possibly non-compact but this can occur only in the case where there exists states $q$ with $\next_{b}(q)<\mnshft(q)$. In this case, Lemma \ref{lemmaComp} may be applied, possibly several times, in order get a compact and valid $w$-matching machine $\matchine_{c}$ with a greater asymptotic speed than $\matchine_{b}$.
Moreover, $\matchine_{c}$ does not contain any state $q$ with $\next_{c}(q)<\mnshft(q)$. 

Let us put $\matchine_{d}$ for $\posit{(\matchine_{c})}$. All the states of $\matchine_{d}$ are relevant. If  $\matchine_{d}$ is not standard and compact, Proposition \ref{propComp} ensures that there exists a $w$-matching machine
$\matchine_{e}$ which is valid, standard, compact, in which all the states are relevant and with a greater asymptotic speed than $\matchine_{d}$.

Finally, applying Corollary \ref{coroDup} on $\matchine_{e}$ as long as there exist two states $q\neq q'$ with $\mem[\matchine_{e}]{q} = \mem[\matchine_{e}]{q'}$, eventually leads  to a $w$-matching machine with the desired properties.
\end{proof}

\begin{corollary}\label{coroFin}
Let $w$ be a pattern, $\piid$ be an iid model and $n$ an integer greater than $\size{w}-1$. Among all the valid $w$-matching machines of order $n$, there exists a machine $\matchine$ with a maximal asymptotic speed which verifies the properties of Theorem \ref{theoMain}. In particular, it is standard, non-redundant and such that $\stateset$ is in bijection with a subset of the partial functions $f$ from $\{0,\ldots, n\}$ to $\alp$, verifying that if $f(i)$ is defined and $i<\size{w}$ then $f(i) = w_{i}$.
\end{corollary}
\begin{proof}
With Theorem \ref{theoMain}, a valid $w$-matching machine of order $n$ achieving the greatest asymptotic speed among the machines of order $n$, may be found among the $w$-matching machines $\matchine$ which are, among other properties, 
\begin{enumerate}
	\item standard,
	\item such that there is no pair of states $(q,q')$ such that $q\neq q'$ and $\mem[\matchine_{\piid}]{q} = \mem[\matchine_{\piid}]{q'}$,
	\item in which all the states are relevant.
\end{enumerate}
The first property just ensures that the second one makes sense. The second property implies the bijection between the set of states and a subset of partial functions from $\{0,\ldots, n\}$ to $\alp$ by associating the states $q$ with the partial function $f_{q}$ corresponding to $\mem[\matchine]{q}$. If for a state $q$ and a position $i<\size{w}$, we have $f(i) \neq w_{i}$, then the state $q$ is not relevant (or the $\matchine$ is not valid), which contradicts the property 3 (or the validity of $\matchine$).
\end{proof}

\begin{corollary}\label{coroFinbis}
Let $w$ be a pattern, $\piid$ be an iid model and $n$ an integer greater than $\size{w}-1$. A $w$-matching machine achieving the greatest asymptotic speed among all the $w$-matching machines of order smaller than $n$ can be computed in a finite time and with a finite amount of memory.
\end{corollary}
\begin{proof}
We first remark that, from Theorem \ref{theoStdValid}, checking if a given $w$-matching machine is valid can be performed in a finite time. Computing its asymptotic speed with regard to an iid model $\piid$ just needs to determine the limit frequencies of a finite Markov chain, which can also be done in a finite time.

Finally, since the subsets of partial functions from $\{0,\ldots, n\}$ to $\alp$ is finite, checking the validity and computing the asymptotic speeds with regard to an iid model $\piid$ of all the $w$-matching machines of which the set of states is in bijection with a partial function from $\{0,\ldots, n\}$ to $\alp$, can be performed in a finite time.
\end{proof}

Being given a pattern $w$, an iid model $\piid$ and an order $n$, it is thus possible to determine with certainty a $w$-matching machine which achieves the greatest asymptotic speed, thus somehow  the smallest asymptotic average complexity on texts following the distribution $\piid$. In the companion paper \cite{didierY}, we provide an algorithm for determining, being given any pattern $w$,  an optimal $w$-matching machine of order $\size{w}-1$ with regard to a given iid model $\piid$ (i.e. with the greatest asymptotic speed under $\piid$). Table \ref{table} displays the asymptotic speeds of some standard algorithms (see \cite{Crochemore1994,Charras2004}) and that of the optimal ``fastest'' one under a given iid model.

\begin{table}
\includegraphics[width=\textwidth]{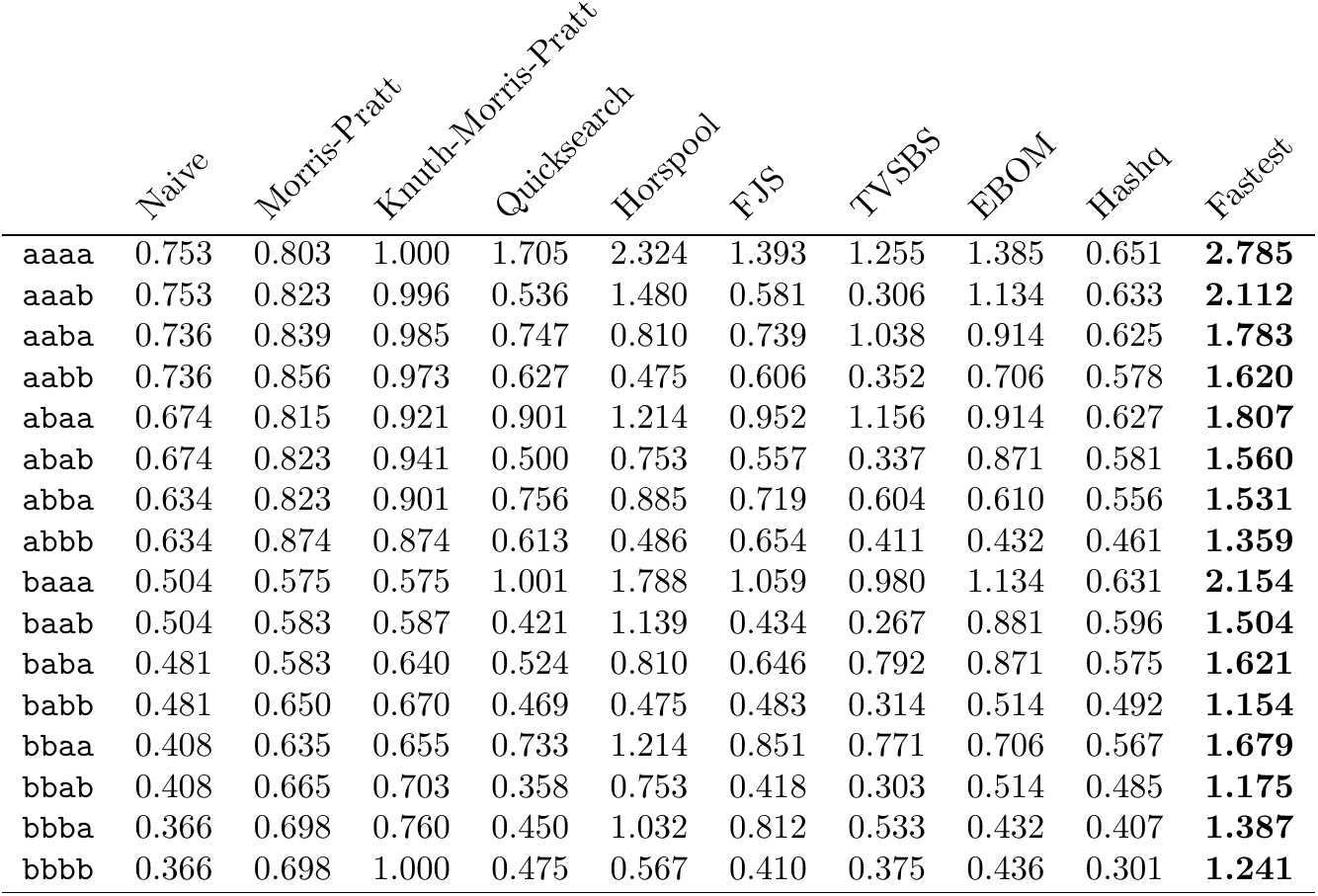}
\caption{Asymptotic speeds of standard algorithms for all the patterns of length $4$ over \{\texttt{a}, \texttt{b}\} with $\piid_{\mbox{\texttt{a}}} = 0.25$ and $\piid_{\mbox{\texttt{b}}} = 0.75$.}\label{table}
\end{table}

\begin{table}
\includegraphics[width=\textwidth]{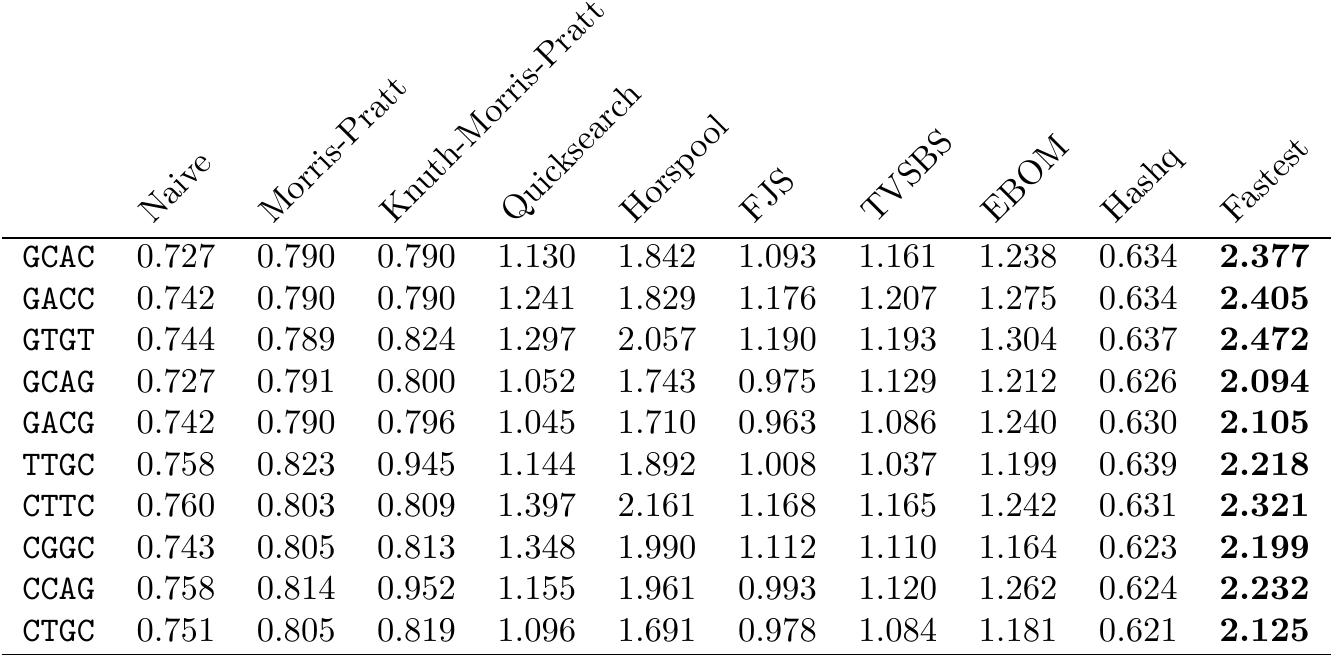}
\caption{Average speeds of standard algorithms over \textit{E. Coli} genome for $10$ patterns randomly picked in the sequence.}\label{tableColi}
\end{table}

\bibliographystyle{abbrv}
\bibliography{PatternMatching}

\end{document}